\documentclass[10pt, twocolumn, comsoc]{IEEEtran}

\usepackage[T1]{fontenc}

\usepackage{graphicx,epsfig}
\usepackage[noadjust]{cite}
\usepackage{mcite}
\usepackage{amsfonts,helvet}
\usepackage{fancyhdr}
\usepackage{threeparttable}
\usepackage{epsf,epsfig}
\usepackage{amsthm}
\usepackage{amsmath}
\usepackage{siunitx}
\usepackage{amssymb}
\usepackage{stfloats}

\usepackage{dsfont}
\usepackage[caption=false,font=footnotesize]{subfig}
\usepackage{color}
\usepackage{enumerate}
\usepackage{gensymb}
\usepackage{cancel}
\usepackage{lipsum}
\usepackage{mathtools}
\usepackage{cuted}
\usepackage{bbm}
\usepackage[linesnumbered,ruled]{algorithm2e}
\usepackage[colorlinks=true, linkcolor=blue]{hyperref}

\usepackage{algorithmic}

\newtheorem{lemma}{Lemma}

\newtheorem{remark}{Remark}

\usepackage{eucal}
\usepackage{booktabs}

\usepackage{multirow}

\begin{document}

\title{Fronthaul Compression for Uplink Cloud-RAN with Finite-Alphabet Inputs:\\ A Reverse Mercury/Waterfilling Approach}

    \author{Subin~Shin, Jaehoon~Lee, Seok-Hwan~Park, and Jeonghun~Park
\thanks{
This work was supported in part by the Institute of Information \& Communications Technology Planning \& Evaluation (IITP) grant funded by the Korea government (MSIT) (No. RS-2024-00395824, Development of Cloud virtualized RAN (vRAN) system supporting upper-midband; No. RS-2024-00404972, Development of 5G-A vRAN Research Platform; No. RS-2024-00435652, 6GARROW: 6G AI-native integrated RAN-Core networks), and in part by the 6GARROW project funded by the Smart Networks and Services Joint Undertaking (SNS JU) under the European Union’s Horizon Europe research and innovation programme (Grant Agreement No. 101192194).
S. Shin, J. Lee and J. Park are with the School of Electrical and Electronic Engineering, Yonsei University, Seoul 03722, South Korea (e-mail:{\texttt{ sbshin@yonsei.ac.kr, jeje5934@gmail.com, jhpark@yonsei.ac.kr}}). 
S.-H. Park is with the School of Electrical Engineering, Hanyang University, Ansan 15588, South Korea (e-mail:{\texttt{ seokhwanpark@hanyang.ac.kr}}).
}
}

\maketitle \setcounter{page}{1} 
\begin{abstract}

The cloud radio access network (C-RAN) mitigates inter-cell interference by jointly processing the observations of distributed remote units (RUs) at a centralized unit (CU), but limited fronthaul capacity forces each RU to compress its received signal. Under transform-compress-forward, an RU transforms its signal and quantizes the resulting coefficients, with bit allocation distributing a finite bit budget across them. Classical reverse waterfilling assumes Gaussian sources, yet practical finite-alphabet symbols carry mutual information that saturates at $\log_2 M$, leaving bit allocation for such inputs unresolved. We address this by formulating bit allocation as maximizing the finite-alphabet generalized mutual information (GMI) achieved after linear MMSE (LMMSE) detection at the CU. Via the I-MMSE relation, this yields a fixed-point update whose converged solution decomposes into a vessel height, a shared water level, and a finite-alphabet mercury level; we term it {reverse mercury/waterfilling} (RMWF). Numerical results show that RMWF sustains end-to-end rate under tight fronthaul budgets and remains robust under antenna scaling, which is increasingly consequential as antenna counts outpace fronthaul capacity in modern C-RAN.
\end{abstract}

\section{Introduction}

Conventional cellular networks process each base station's received signals
independently, treating transmissions from neighboring cells as interference. As networks densify, this inter-cell interference (ICI) becomes a dominant bottleneck. 
The cloud radio access network (C-RAN) architecture~\cite{gesbert2010multi,peng2015fronthaul} addresses this by forwarding observations from distributed remote units (RUs) to a centralized unit (CU), where joint processing mitigates ICI at the network level. 
This cooperative gain, however, is bounded by the fronthaul links between the RUs and the CU: finite capacity forces each RU to compress its observation before forwarding, introducing distortion that degrades performance. 
Efficient fronthaul compression is, therefore, central to realizing the performance gains promised by C-RAN.

From an information-theoretic perspective, C-RAN fronthaul compression has been studied through 
distributed source-coding formulations under Gaussian signaling. 
Specifically, in the uplink, the key idea for approaching the network capacity is Wyner--Ziv coding, in which each RU compresses its observation by exploiting the previously decompressed signals at the CU as side information; the CU then performs joint decoding of the user messages over the aggregated reconstructions~\cite{sanderovich2009uplink, park2013robust, zhou2014optimized}. 
Realizing this scheme requires jointly designed vector-quantized codebooks whose encoding complexity grows exponentially with the block length, making it prohibitive for practical fronthaul, even before the binning structure that Wyner--Ziv coding additionally demands.



To circumvent the complexity of vector quantization, a widely used practical family of fronthaul compression schemes follows the transform-compress-forward paradigm \cite{liu2015optimized, liu2019twotimescale, zhang2021quantization, qiao2024meta}.
In this paradigm, each RU first applies a transform that maps its received signal to a decorrelated, dimension-reduced representation, then quantizes each coefficient via low-complexity scalar quantization under the fronthaul budget. 
Finally, each quantized representation is forwarded to the CU. 
On the transform side, the Karhunen--Lo\`eve transform (KLT) is a suitable choice, as it decorrelates the received signal along its statistical eigenbasis and concentrates energy on a few dominant coefficients \cite{liu2015optimized, wiffen2020dimension, wiffen2021distributed}.


Once the KLT is adopted as the transform, the remaining crucial design question is how to allocate a finite bit budget across the coefficients. This allocation directly governs how the limited fronthaul capacity is translated into end-to-end rate.
For Gaussian sources, a closely related question of sum-distortion minimization under a total bit budget admits a classical solution via reverse waterfilling (RWF)~\cite{CoverThomas}, which assigns more bits to high-variance coefficients and fewer bits to low-variance ones, according to the logarithm of the source variances.
This scheme, however, relies on the Gaussian source assumption, which does not hold in modern communication systems. In actual deployments, user data is drawn from finite constellations such as $M$-QAM, so what each RU forwards is a noisy observation of finite-alphabet symbols rather than of a Gaussian source. This makes a clear difference at the bit allocation stage: Gaussian mutual information (MI) grows logarithmically with the signal-to-noise ratio (SNR) without bound, whereas finite-alphabet MI saturates at $\log_2 M$ once the constellation is reliably resolved.

A parallel Gaussian-vs-finite-alphabet question has, in fact, been resolved on the power allocation side. Classical waterfilling (WF) distributes transmit power across parallel Gaussian channels~\cite{CoverThomas}, and for finite-alphabet inputs, mercury/waterfilling (MWF) modifies this scheme by injecting a constellation-dependent mercury level that captures the finite-alphabet MI saturation~\cite{LozanoTulinoVerdu}. On the bit allocation side, however, a finite-alphabet counterpart that captures the MI saturation remains an open problem, even though this is precisely the operating regime of C-RAN fronthaul compression with finite-alphabet signaling. Motivated by this missing corner, this work proposes a bit allocation scheme that explicitly accounts for finite-alphabet inputs.

\subsection{Related Works}
\label{sec:related}

{\bf{Gaussian Compress-and-Forward:}}
A foundational line of work characterizes the achievable rates over finite-capacity RU-CU links under Gaussian signaling; a comprehensive treatment of the underlying signal-processing and network-information-theoretic principles is provided in~\cite{park2014fronthaul}.
The uplink direction was first studied through compress-and-forward formulations~\cite{sanderovich2009uplink} and subsequently refined by Wyner--Ziv style robust compression~\cite{park2013robust}, quantization noise covariance design~\cite{zhou2014optimized}, and optimal joint compression-decoding strategies~\cite{zhou2016optimal}. The downlink direction was developed through joint precoding with multivariate compression~\cite{park2013joint} and later through Marton-based generalized compression~\cite{patil2019generalized}, with uplink-downlink duality established in~\cite{liu2021duality}.

Building on the same Gaussian abstraction, subsequent works have embedded fronthaul compression into broader network-optimization problems, including joint beamforming and quantization noise covariance design for the uplink~\cite{zhou2016fronthaul}, hybrid data-sharing/compression for the downlink~\cite{patil2018hybrid}, and sparse joint transmission under limited fronthaul capacity~\cite{han2022sparse}. These works treat the bit budget at the abstraction level of test channels, covariance matrices, or fronthaul capacity variables under Gaussian signaling, and an explicit bit allocation scheme for finite-alphabet signals does not arise within this line.

{\bf{Transform-Compress-Forward and Bit Allocation:}}
A separate line of work makes the transform and scalar-quantization stages explicit, with the bit allocation across transformed coefficients emerging as the central design variable. 
A representative approach is spatial-compression-and-forward for multi-antenna RUs, where each RU applies a linear spatial filter across its antenna array followed by per-coefficient uniform scalar quantization, with the filters designed locally and the transmit powers, receive beamformers, and bit allocation jointly optimized at the CU~\cite{liu2015optimized}. The KLT adopted in this paper can be viewed as a natural realization of such a decorrelating spatial filter.



Focusing on the bit allocation aspect, several variants of this paradigm have been investigated: mixed-resolution analog-to-digital converter (ADC) architectures that assign different resolutions across antenna chains under a sum resolution constraint guided by low-SNR rate approximations~\cite{park2017mixed}; distributed compressive sensing that reduces the RU-side dimension and recovers the signal jointly at the CU~\cite{rao2015distributed}; quantization-aware processing that concentrates resources on informative transformed coefficients~\cite{zhang2021quantization}; and two-timescale hybrid compress-and-forward separating long-term analog filtering from short-term digital filtering and bit allocation~\cite{liu2019twotimescale}. However, these works all rely on the Gaussian source assumption and do not capture the saturating MI behavior of finite-alphabet sources.


{\bf{Learning-Based Transform:}}
In parallel, learning-based approaches have enlarged the design space beyond analytically designed transforms. For instance, deep learning has been adopted as a low-complexity surrogate for joint beamforming and fronthaul quantization design~\cite{yu2021deep}, while progressive distributed compression has been learned from local channel state information (CSI) with a mean squared error (MSE) reconstruction objective~\cite{sohrabi2022learning}.


More directly in C-RAN fronthaul compression, meta-learning generates RU-side linear transformation matrices from local CSI, with a gated recurrent unit (GRU)-based refinement stage that uses low-dimensional global feedback to improve sum rate while retaining Gaussian signaling assumptions~\cite{qiao2024meta}.
A related line is neural source compression, which jointly learns the transform, quantizer, and entropy model under a rate-distortion objective~\cite{balle2018variational}.
This idea has been adapted to neural fronthaul compression for modulated waveforms, adopting an error vector magnitude (EVM)-based distortion measure rather than optimizing communication MI~\cite{bian2025towards}.
Closest to the present problem, \cite{chene2024distributed} jointly learns a linear filter, a uniform scalar quantizer with non-homogeneous bit allocation, and a decoder for uplink C-RAN under finite-alphabet signaling and a per-RU bit budget; however, the bit allocation emerges implicitly from black-box supervised training and a bit budget penalty rather than from an interpretable information-theoretic rule.

Across these three lines, the same gap remains: there is no explicit and interpretable fronthaul compression framework for finite-alphabet signaling that captures the saturating MI behavior of finite-alphabet sources.
This paper addresses this gap by directly maximizing the post-quantization MI at the CU under sum or per-RU fronthaul constraints.
Our contribution is therefore not a new transform, quantizer, or entropy model, but an interpretable bit allocation scheme designed specifically to capture this finite-alphabet saturation.


\subsection{Contributions}
\label{sec:contributions}


In this paper, we develop a fronthaul compression method built on the KLT followed by our proposed reverse mercury/waterfilling (RMWF) bit allocation, designed specifically for finite-alphabet inputs in the C-RAN uplink. Our main contributions are as follows.



{\bf{A General Bit Allocation Framework for Finite-Alphabet C-RAN
Uplinks:}}
We formulate the fronthaul bit allocation problem under finite-alphabet inputs as the maximization of the post-linear minimum mean squared error (LMMSE) finite-alphabet generalized mutual information (GMI)~\eqref{eq:gmi} at the CU.
Two operationally distinct fronthaul provisioning regimes, namely a sum constraint modeling pooled provisioning, and a per-RU constraint modeling dedicated per-link budgets, are handled within a single algorithmic structure. 
Prior C-RAN bit allocation works generally operate under the Gaussian rate-distortion framework~\cite{liu2015optimized, zhang2021quantization}, use low-SNR GMI approximations that miss the $\log_2 M$ saturation~\cite{park2017mixed}, or pursue finite-alphabet allocation through black-box learned surrogates~\cite{chene2024distributed}. 
In contrast, the proposed framework adopts the finite-alphabet, MI-aware view as the explicit system objective.

{\bf{RMWF Structure with a Four-Way Correspondence:}}
Applying the I-MMSE relation~\cite{GuoShamaiVerdu} on the bit allocation side of the chain, we derive a fixed-point update whose converged allocation decomposes into a per-coefficient vessel height, a shared water level, and a finite-alphabet mercury level, generalizing the RWF structure through an explicit saturation correction. 
In the Gaussian high-SNR regime, the allocation recovers the same RWF structure (classically derived for rate-distortion~\cite{CoverThomas}); together with WF and MWF on the power allocation side~\cite{LozanoTulinoVerdu}, this establishes a unified MI-maximizing view of the four-way correspondence (Table~\ref{tab:correspondence}), with the finite-alphabet bit allocation corner being the principal new derivation.

{\bf{Numerical Validation under Realistic C-RAN Deployments:}}
The framework is validated on a 3GPP 3D urban-microcell (UMi) C-RAN with hexagonal layouts and mixed modulation orders drawn from $\{4,16,64,256\}$-QAM. Three findings are consistent across all tested configurations: the proposed allocator substantially outperforms both Equal-Bit and Gaussian-input allocations under tight fronthaul budgets; it approaches the unquantized upper bound as the budget grows; and under antenna scaling at fixed fronthaul, it remains on a near-flat plateau while competing schemes degrade. This robustness hierarchy becomes consequential as antenna counts scale faster than fronthaul capacity. The gains carry over from system-level sum GMI to the per-user GMI distribution.

\section{System Model and Problem Formulation}
\label{sec:system_model}

\subsection{C-RAN Uplink with Finite-Alphabet Inputs Model}\label{subsec:cran_fa}

As illustrated in Fig.~\ref{fig:sys}, we consider a C-RAN uplink in which $K$ single-antenna 
users are served by a CU\footnote{We follow the O-RAN terminology \cite{oran2020wp}, where the RU performs RF and low-PHY processing while the CU performs centralized baseband processing. Throughout this paper, we adopt a fully centralized split, under which the high-PHY functions traditionally handled by the distributed unit (DU) are absorbed into the CU.} through $M$ RUs, 
each equipped with $N$ antennas. Each user $k \in 
\CMcal{K} = \{1, \dots, K\}$ is assigned a modulation 
order $M_k \in \{4, 16, 64, 256\}$ corresponding 
to 4-QAM, 16-QAM, 64-QAM, or 256-QAM, and transmits 
with power $P_k$. Let $x_k \in \CMcal{X}_{M_k} 
\subset \mathbb{C}$ denote the symbol transmitted by user $k$, drawn uniformly from the constellation 
$\CMcal{X}_{M_k}$ with 
$\mathbb{E}[|x_k|^2] = P_k$.

The signal observed by the RU $m$ is
\begin{align} \label{eq:received_signal}
  \mathbf{y}_m 
  = \sum_{k=1}^{K} \mathbf{h}_{m,k} x_k 
    + \mathbf{n}_m 
  = \mathbf{H}_m \mathbf{x} + \mathbf{n}_m 
  \in \mathbb{C}^{N},
\end{align}
where $\mathbf{h}_{m,k} \in \mathbb{C}^N$ is the 
channel vector from user $k$ to RU $m$, 
$\mathbf{n}_m \sim \mathcal{CN}(\mathbf{0}, \sigma^2 \mathbf{I}_N)$ is additive Gaussian noise, and the 
aggregated quantities are 
$\mathbf{x} = [x_1, \dots, x_K]^{\sf T} \in 
\mathbb{C}^K$ and $\mathbf{H}_m = [\mathbf{h}_{m,1}, 
\dots, \mathbf{h}_{m,K}] \in \mathbb{C}^{N \times K}$.

Each RU compresses its received signal $\mathbf{y}_m$ into a finite-rate bit stream before forwarding it to the CU.
Following standard practice~\cite{park2014fronthaul}, we adopt a transform coding architecture in which the received signal is first projected onto its statistical eigenbasis via the KLT, and the resulting decorrelated coefficients are then independently quantized. Specifically, let $\mathbf{R}_{\mathbf{y}_m} \triangleq \mathbb{E}[\mathbf{y}_m \mathbf{y}_m^{\sf H}] \in \mathbb{C}^{N \times N}$ denote the received signal covariance matrix at RU $m$, with eigenvalue decomposition
\begin{align}
  \mathbf{R}_{\mathbf{y}_m} = \mathbf{U}_m 
  \boldsymbol{\Lambda}_m \mathbf{U}_m^{\sf H},
\end{align}
where $\mathbf{U}_m \in \mathbb{C}^{N \times N}$ is unitary and $\boldsymbol{\Lambda}_m = \mathrm{diag}(\lambda_{m,1}, \dots, \lambda_{m,N})$ collects the eigenvalues in non-increasing order. 
The KLT coefficients are
\begin{align} \label{eq:klt}
  \boldsymbol{z}_m 
  \triangleq \mathbf{U}_m^{\sf H} \mathbf{y}_m 
  = [z_{m,1}, \dots, z_{m,N}]^{\sf T} 
  \in \mathbb{C}^{N},
\end{align}
which are mutually uncorrelated with $\mathbb{E}[|z_{m,n}|^2] = \lambda_{m,n}$.


Each coefficient $z_{m,n}$ is then quantized independently with $b_{m,n}$ bits per complex scalar. We adopt subtractive dithered quantization~\cite{schuchman1964dither}: the encoder at RU~$m$ adds a dither $\delta_{m,n}$ to $z_{m,n}$, applies a uniform scalar quantizer $Q_{m,n}(\cdot)$, and forwards the resulting quantization index to the CU over the fronthaul. 
The CU regenerates the corresponding reproduction level from the index and subtracts the dither realization, yielding
\begin{align}\label{eq:dithered_quant}
  \hat{z}_{m,n} = Q_{m,n}\!\left(z_{m,n} + \delta_{m,n}\right) - \delta_{m,n},
\end{align}
where $Q_{m,n}(\cdot)$ is a uniform quantizer with cell width $\Delta_{m,n}$ applied separately to the in-phase and quadrature components, and the dither $\delta_{m,n}$ has real and imaginary parts drawn independently from $\mathrm{Unif}\left([-\Delta_{m,n}/2,\,\Delta_{m,n}/2)\right)$, independent of $\mathbf{x}$ and $\{\mathbf{n}_m\}$.
The dither realizations are generated from a pseudorandom seed shared between RU~$m$ and the CU as common randomness, with no per-sample fronthaul overhead.

Defining the quantization error as $e_{m,n} \triangleq \hat{z}_{m,n} - z_{m,n}$, the construction~\eqref{eq:dithered_quant} yields the additive decomposition
\begin{align} \label{eq:additive_decomp}
  \hat{z}_{m,n} = z_{m,n} + e_{m,n}.
\end{align}
Under~\eqref{eq:dithered_quant}, the Schuchman condition~\cite{schuchman1964dither} guarantees that, in the granular (non-overload) regime, $e_{m,n}$ is statistically independent of $z_{m,n}$ and uniformly distributed over $[-\Delta_{m,n}/2,\,\Delta_{m,n}/2)$ on each real dimension, for any source distribution of $z_{m,n}$ and any bit allocation $b_{m,n}$.

Tying the cell width to the coefficient dynamic range via a loading factor $\kappa$, with the per-real-dimension clipping range $A_{m,n} = \kappa\sigma_{r,m,n}$ and $\sigma_{r,m,n}^2 = \lambda_{m,n}/2$ denoting the per-real-dimension variance of $z_{m,n}$, the cell width is
\begin{align} \label{eq:cell_width}
  \Delta_{m,n} = \frac{2 A_{m,n}}{2^{b_{m,n}/2}} = \frac{2\kappa\sigma_{r,m,n}}{2^{b_{m,n}/2}}.
\end{align}
Since the complex error $e_{m,n}$ aggregates the in-phase and quadrature components, each uniform over a cell of width $\Delta_{m,n}$, its second moment admits the following form
\begin{align} \label{eq:qmn_def}
  q_{m,n} 
  \triangleq \mathrm{Var}(e_{m,n})
  = 2\cdot\frac{\Delta_{m,n}^2}{12}
  = c \, \lambda_{m,n} \cdot 2^{-b_{m,n}},
\end{align}
where $c = \kappa^2/3$ is the dithered-uniform quantizer constant; for instance, a $3\sigma$ loading ($\kappa = 3$) yields $c = 3$. The exponential factor $2^{-b_{m,n}}$ reflects the classical distortion-rate scaling of uniform quantization~\cite{bennett1948spectra, gishpierce1968asymptotically, GershoGray}. 

Stacking the quantized coefficients $\hat{\boldsymbol{z}}_m = 
\boldsymbol{z}_m + \mathbf{e}_m$ and applying the inverse transform, the reconstructed signal at the CU is
\begin{align} \label{eq:reconstruction}
  \hat{\mathbf{y}}_m 
  = \mathbf{U}_m \hat{\boldsymbol{z}}_m 
  = \mathbf{y}_m + \mathbf{d}_m,
\end{align}
where $\mathbf{d}_m \triangleq \mathbf{U}_m \mathbf{e}_m$ is the equivalent quantization noise in the original signal domain. Its covariance is 
$\mathbb{E}[\mathbf{d}_m \mathbf{d}_m^{\sf H}] = 
\mathbf{U}_m \mathrm{diag}(q_{m,1}, \dots, q_{m,N}) \mathbf{U}_m^{\sf H}$, which is fully determined by the bit allocation $\{b_{m,n}\}_{n=1}^{N}$ at RU $m$.

\begin{figure*}[t]     
\centerline{\resizebox{1.8\columnwidth}{!}{\includegraphics{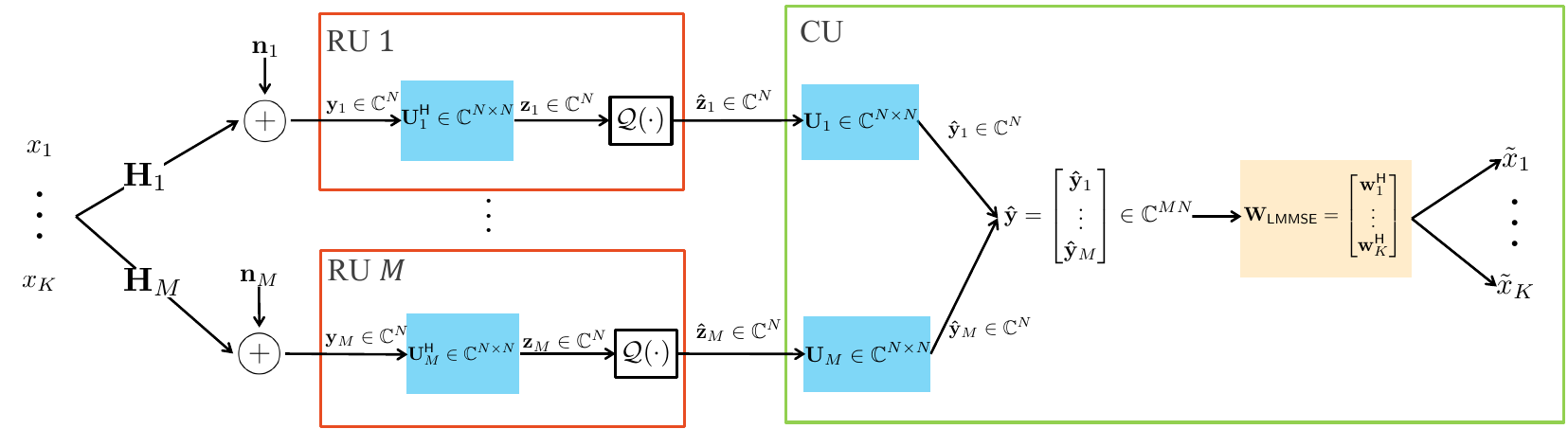}}}
     \caption{System model of the uplink C-RAN. $K$ users transmit through $\{\mathbf{H}_m\}$ to $M$ RUs with $N$ antennas each. Each RU $m$ projects $\mathbf{y}_m$ via the KLT $\mathbf{U}_m^{\sf H}$, quantizes the coefficients with bits $\{b_{m,n}\}$, and forwards to the CU. The CU applies the inverse transforms $\{\mathbf{U}_m\}$, stacks into $\hat{\mathbf{y}}\in\mathbb{C}^{MN}$, and recovers the $K$ symbols via the LMMSE receiver $\mathbf{W}_{\sf LMMSE}$.}
     \label{fig:sys}
\end{figure*}

\subsection{Joint Decoding Model and Performance Metric}\label{subsec:joint_decoding}
The CU collects the reconstructed signals $\{\hat{\mathbf{y}}_m\}_{m=1}^{M}$ from all RUs and stacks them into the aggregated observation
\begin{align} \label{eq:vmac}
  \hat{\mathbf{y}} 
  \triangleq \big[\, \hat{\mathbf{y}}_1^{\sf T},\, 
       \dots,\, \hat{\mathbf{y}}_M^{\sf T} \,\big]^{\sf T} 
  = \mathbf{H} \mathbf{x} + \mathbf{n} + \mathbf{d} 
  \in \mathbb{C}^{MN},
\end{align}
where $\mathbf{H} = [\mathbf{H}_1^{\sf T}, \dots, \mathbf{H}_M^{\sf T}]^{\sf T} \in \mathbb{C}^{MN \times K}$ is the aggregated channel matrix, $\mathbf{n} = [\mathbf{n}_1^{\sf T}, \dots, \mathbf{n}_M^{\sf T}]^{\sf T} \sim \mathcal{CN}(\mathbf{0}, \sigma^2 \mathbf{I}_{MN})$ is the stacked Gaussian noise, and $\mathbf{d} = [\mathbf{d}_1^{\sf T}, \dots, \mathbf{d}_M^{\sf T}]^{\sf T}$ is the stacked quantization noise with block-diagonal covariance
\begin{align} \label{eq:Q_cov}
  \mathbf{Q} 
  \triangleq \mathbb{E}[\mathbf{d} \mathbf{d}^{\sf H}] 
  = \mathrm{blkdiag}\left(
       \mathbf{U}_1 \mathbf{Q}_1 \mathbf{U}_1^{\sf H},\, 
       \dots,\, 
       \mathbf{U}_M \mathbf{Q}_M \mathbf{U}_M^{\sf H}
     \right),
\end{align}
where $\mathbf{Q}_m \triangleq \mathrm{diag}(q_{m,1}, \dots, q_{m,N})$. Under the dithered construction~\eqref{eq:dithered_quant}, $\mathbf{d}$ is zero-mean and statistically independent of $\mathbf{x}$, so the aggregate disturbance has covariance $\sigma^2\mathbf{I}_{MN}+\mathbf{Q}$.

The CU recovers user symbols by applying a linear receiver to $\hat{\mathbf{y}}$.
Specifically, we adopt the LMMSE receiver given by
\begin{align} \label{eq:lmmse}
  \mathbf{W}_{\sf LMMSE} 
  = \mathbf{P} \mathbf{H}^{\sf H} \mathbf{R}_{\mathbf{y}}^{-1},
\end{align}
where $\mathbf{P} \triangleq \mathrm{diag}(P_1, \dots, P_K)$ and $\mathbf{R}_{\mathbf{y}} \triangleq \mathbf{H} \mathbf{P} \mathbf{H}^{\sf H} + \sigma^2 \mathbf{I}_{MN} + \mathbf{Q}$. 
The soft estimate for user $k$ is then $\tilde{x}_k = [\mathbf{W}_{\sf LMMSE} \hat{\mathbf{y}}]_k$.

The post-equalization SINR for user $k$ admits the standard expression
\begin{align} \label{eq:sinr}
  \mathrm{SINR}_k 
  &= \frac{P_k \, |\mathbf{w}_k^{\sf H} \mathbf{h}_k|^2}
         {\mathbf{w}_k^{\sf H} \big(\sum_{j \neq k} 
            P_j \mathbf{h}_j \mathbf{h}_j^{\sf H} 
            + \sigma^2 \mathbf{I}_{MN} + \mathbf{Q}\big) 
          \mathbf{w}_k},
\end{align}
where $\mathbf{h}_k$ is the $k$-th column of $\mathbf{H}$ and $\mathbf{w}_k^{\sf H}$ is the $k$-th row of $\mathbf{W}_{\sf LMMSE}$. Defining the interference-plus-noise covariance $\mathbf{R}_k \triangleq \sum_{j \neq k} P_j \mathbf{h}_j \mathbf{h}_j^{\sf H} + \sigma^2 \mathbf{I}_{MN} + \mathbf{Q}$ so that $\mathbf{R}_{\mathbf{y}} = \mathbf{R}_k + P_k \mathbf{h}_k\mathbf{h}_k^{\sf H}$, the matrix inversion lemma applied to $\mathbf{R}_{\mathbf{y}}$ in~\eqref{eq:sinr} yields the equivalent compact form $P_k \mathbf{h}_k^{\sf H} \mathbf{R}_k^{-1} \mathbf{h}_k$.

The LMMSE output admits the decomposition $\tilde{x}_k = \rho_k x_k + \eta_k$, where $\rho_k \triangleq \mathbf{w}_k^{\sf H}\mathbf{h}_k$ and $\eta_k$ collects the residual multi-user interference, noise, and quantization distortion. Under the dithered construction, $\eta_k$ is independent of the desired symbol $x_k$, but remains non-Gaussian due to the finite-alphabet multi-user interference.
The decoder applies a Gaussian nearest-neighbor metric, which is mismatched to the true $\eta_k$ and yields the GMI~\cite{merhav1994information, ganti2000mismatched}:
\begin{align} \label{eq:gmi}
  I_k^{\sf GMI} 
  = \sup_{s \ge 0}\;
    \mathbb{E} \!\left[
      \log_2 
      \frac{\exp \!\big(-s\,|\tilde{x}_k - \rho_k x_k|^2\big)}
           {\frac{1}{M_k}\sum_{x' \in \CMcal{X}_{M_k}} 
             \exp \!\big(-s\,|\tilde{x}_k - \rho_k x'|^2\big)}
    \right],
\end{align}
where the expectation is over the joint law of $(x_k, \tilde{x}_k)$ and $M_k$ is the modulation order of user $k$. 
The GMI in~\eqref{eq:gmi} is a lower bound on $I(x_k; \tilde{x}_k)$.

Setting $\gamma_k \triangleq \mathrm{SINR}_k$, we define the constellation-constrained mutual information (CCMI) of a unit-power finite-alphabet input over a scalar AWGN channel~\cite{bruno2026mismatch}:
\begin{align} \label{eq:mi_finite_def}
  I_{M_k}(\gamma_k) 
  &= \log_2 M_k \nonumber \\
  &\quad - \mathbb{E}_{x, n}  \left[
    \log_2 \sum_{x' \in \CMcal{X}_{M_k}} 
    \exp  \left(
      -\frac{|x - x' + n|^2 - |n|^2}{1/\gamma_k}
    \right)
  \right],
\end{align}
with $x$ uniform on $\CMcal{X}_{M_k}$ and $n \sim \mathcal{CN}(0, 1/\gamma_k)$. 
We adopt the CCMI $I_{M_k}(\gamma_k)$ as the system-level design objective throughout this paper; Remark~\ref{rem:gaussian_residual} details its operational meaning.

\begin{remark}[Operational meaning of the GMI under the dithered residual] \label{rem:gaussian_residual}
\normalfont
By the mismatched-decoding theorem~\cite{merhav1994information, ganti2000mismatched},
the GMI in~\eqref{eq:gmi} is an unconditional lower bound on $I(x_k;\tilde{x}_k)$, requiring no structural assumption on $\eta_k$. 
As computing $I_k^{\sf GMI}$ within the bit-allocation loop is costly, we instead optimize the CCMI surrogate $I_{M_k}(\gamma_k)$ in~\eqref{eq:mi_finite_def}, which replaces $\eta_k$ by an independent Gaussian of matched variance. 
Although a worst-case-Gaussian guarantee fails for general finite-alphabet inputs~\cite{carmon2012disproof}, the known counterexamples require large skewness or excess kurtosis.
In our case, $\eta_k$ stays close to Gaussian because it is a sum of many finite-alphabet interferers with symmetric, light-tailed quantization errors induced by subtractive dithering and Gaussian noise. 
Therefore, the surrogate is expected to be conservative~\cite{bruno2026mismatch}. 
A direct evaluation indeed finds $I_k^{\sf GMI}$ and $I_{M_k}(\gamma_k)$ virtually indistinguishable across all
tested points (Section~\ref{subsec:ccmi}), making CCMI maximization empirically equivalent to GMI maximization here. This also mirrors commercial practice: LTE/5G NR receivers
compute Gaussian LLRs over post-equalization residuals under real-time constraints.
\end{remark}

\subsection{Problem Formulation}
\label{sec:problem}


We seek to find the bit allocation $\{b_{m,n}\}$, i.e., the number of bits quantizing the $n$-th KLT coefficient at RU $m$, for maximizing the sum GMI across all $K$ users at the CU. The allocation enters this objective via $\mathbf{Q}$ in~\eqref{eq:Q_cov}, hence $\mathrm{SINR}_k$~\eqref{eq:sinr} and $I_k^{\sf GMI}$~\eqref{eq:gmi}; we consider two formulations reflecting different fronthaul capacity regimes.

{\textbf{{Sum Fronthaul Bit Allocation:}}}
In some C-RAN deployments, multiple co-located RUs  share a common high-capacity fronthaul link to the CU. 
For example, this occurs when RUs within the same cell site are aggregated through a passive optical network or a shared Ethernet uplink~\cite{checko2015cloud, peng2015fronthaul}. Under such pooled provisioning, the operator may allocate the total bit budget $C_{\sf tot}$ freely across all RUs and all KLT coefficients, which is formulated as
\begin{align} \label{eq:problem_sum}
  \mathrm{(P1):} \quad 
  \max_{\{b_{m,n}\}} 
  \quad & \sum_{k=1}^{K} 
        I_{M_k} \big(\mathrm{SINR}_k(\{b_{m,n}\})\big) 
        \nonumber \\
  \mathrm{s.t.} 
  \quad & \sum_{m=1}^{M} \sum_{n=1}^{N} b_{m,n} 
        \leq C_{\sf tot}, \\
  & b_{m,n} \geq 0, 
        \quad \forall\, m, n. \nonumber
\end{align}
This setting offers flexibility in bit allocation, since quantization bits can be redistributed both across RUs and across KLT coefficients to maximize the sum GMI.


{\textbf{{Per-RU Fronthaul Bit Allocation:}}}
In typical deployments, each RU is connected to the 
CU via its own dedicated fronthaul link of capacity 
$C_m$, reflecting independent physical link budgets, service level agreement (SLA) isolation, or distinct fronthaul technologies across RUs.
In this case, the bits at each RU must satisfy a local budget:
\begin{align} \label{eq:problem_per_ru}
  \mathrm{(P2):} \quad 
  \max_{\{b_{m,n}\}} 
  \quad & \sum_{k=1}^{K} 
        I_{M_k} \big(\mathrm{SINR}_k(\{b_{m,n}\})\big) 
        \nonumber \\
  \mathrm{s.t.} 
  \quad & \sum_{n=1}^{N} b_{m,n} \leq C_m, 
        \quad m = 1, \dots, M, \\
  & b_{m,n} \geq 0, 
        \quad \forall\, m, n. \nonumber
\end{align}
Problem~(P2) is structurally harder than~(P1).
First, in contrast to~(P1), which involves only a single coupling constraint, (P2) imposes a separate fronthaul capacity constraint for each of the 
$M$ RU-CU links.
Second, bits can no longer be moved across RUs, which removes the main source of flexibility that~(P1) exploits. As a result,~(P2) reduces to $M$ per-RU subproblems that are nonetheless coupled through the joint SINR at the CU, and this coupling is what makes~(P2) difficult to solve.



Both (P1) and (P2) inherit three difficulties:
$I_M(\cdot)$ has no closed form and saturates at $\log_2 M$, breaking the unbounded-growth assumption behind classical RWF; $\mathrm{SINR}_k$ couples $\{b_{m,n}\}$ across all RUs and coefficients through the joint receiver; and the objective is non-convex in general. To address these, the next section studies an illustrative $N$-parallel channel example that builds the core mechanism of our approach, termed RMWF.

\section{Intuitive Example: RMWF} \label{sec:toy_fa}

In this section, we illustrate the key idea behind RMWF through a simple $N$-parallel channel example. 
We first develop the generalized RMWF solution for practical finite-alphabet inputs such as QAM, and then show that for the special case of Gaussian inputs, RMWF reduces to the classical RWF form in the high-SNR regime.

Consider $N$ parallel channels
\begin{align}
  y_i = x_i + n_i, \quad i = 1,\ldots, N
\end{align}
where $n_i \sim \mathcal{CN}(0, \sigma^2)$ is additive Gaussian noise. Each received signal $y_i$ is quantized 
with $b_i$ bits and forwarded to a remote decoder, subject to a total bit budget $\sum_{i=1}^N b_i = B$. 
A relevant question is: how should the bit budget $B$ be distributed across the $N$ channels?

Suppose that each $x_i \in \CMcal{X}_{M_i}$ is drawn from an $M_i$-QAM constellation with $\mathbb{E}[|x_i|^2] = P_i$ for $i = 1, \ldots, N$.
Assuming $b_i$ bits per complex dimension, the quantized observation $\hat{y}_i$ is modeled as $\hat{y}_i = y_i + e_i$ following \eqref{eq:additive_decomp}. 
Also following our quantization model, the quantization error $e_i$ is zero-mean. Its variance is given by
\begin{align} \label{eq:qi}
  q_i \triangleq \mathrm{Var}(e_i) 
  = c \, \sigma_{y_i}^2 \, 2^{-b_i},
\end{align}
where $\sigma_{y_i}^2 = {\rm{Var}}(y_i) = P_i + \sigma^2$. 
Assuming the decoder uses an MMSE estimator, the resulting post-quantization SNR for estimating $x_i$ is given by $\mathrm{SNR}_i = {P_i}/({\sigma^2 + q_i})$. With the post-quantization SNR, we now turn to the system-level optimization. 
Under a total quantization bit budget $B$, we consider the problem of maximizing the sum GMI between the inputs and the quantized observations:
\begin{align} \label{eq:opt_fa}
  \max_{\{b_i\}} \; \sum_{i=1}^N I_{M_i}(\mathrm{SNR}_i) 
  \quad 
  \text{s.t.} \quad \sum_{i=1}^N b_i \leq B, \; b_i \ge 0, \; \forall i.
\end{align} 
Forming the Lagrangian of~\eqref{eq:opt_fa} with multiplier $\mu > 0$, the KKT necessary condition reads $\partial I_{M_i}(\mathrm{SNR}_i)/\partial b_i \le \mu$ with equality when $b_i > 0$.
By the I-MMSE relation~\cite{GuoShamaiVerdu}, which is applicable under the Gaussian decoding metric of Remark~\ref{rem:gaussian_residual}, the marginal information gain in $\mathrm{SNR}$ admits $dI_{M_i}/d\mathrm{SNR}_i = \mathrm{mmse}_{M_i}(\mathrm{SNR}_i)/\ln 2$, where $\mathrm{mmse}_{M_i}(\cdot)$ is the MMSE function of the $M_i$-QAM constellation. Combined with the chain rule applied to $\mathrm{SNR}_i = P_i/(\sigma^2 + q_i)$, together with $\partial q_i / \partial b_i = -(\ln 2)\, q_i$~\eqref{eq:qi}, the KKT condition reduces to
\begin{align} \label{eq:kkt_fa_exact}
  \frac{P_i \, \mathrm{mmse}_{M_i}(\mathrm{SNR}_i)}{(\sigma^2 + q_i)^2} \cdot q_i 
  \begin{cases}
       = \mu & \text{if } b_i > 0 \\
       \leq \mu & \text{if } b_i = 0
  \end{cases}\,,
\end{align}
where $\mu$ is chosen to satisfy the budget constraint.
Unfortunately, \eqref{eq:kkt_fa_exact} admits no closed-form solution in $b_i$ because $\mathrm{mmse}_{M_i}(\cdot)$ has no analytic expression for finite-alphabet inputs. 
However, fixing the saturation factor $P_i\,\mathrm{mmse}_{M_i}(\mathrm{SNR}_i)/(\sigma^2+q_i)^2$ at the current value reduces~\eqref{eq:kkt_fa_exact} to a residual condition in $q_i$ that admits a closed-form solution. 
This naturally suggests a fixed-point iteration: our RMWF alternates between updating the saturation factor from the current allocation and updating the allocation in closed form, until convergence.

\begin{figure}[t]
    \centering
    \includegraphics[width=0.8\columnwidth]{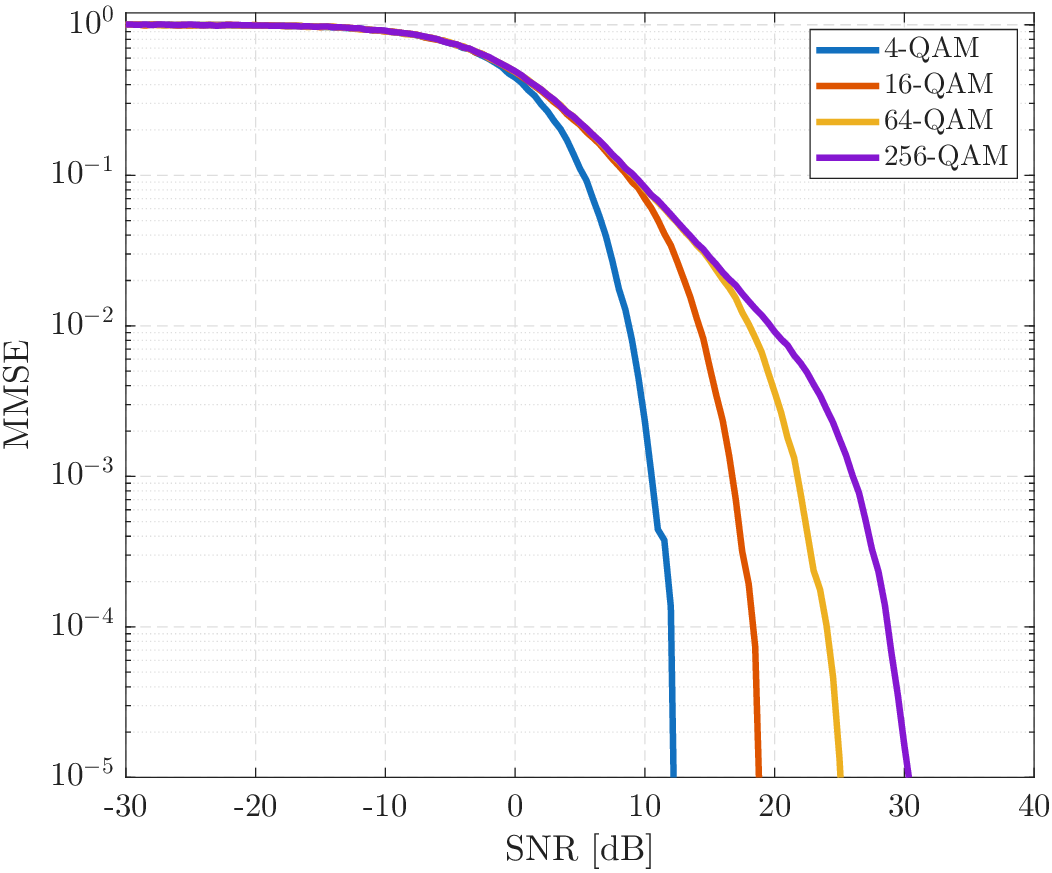}
    \caption{
    MMSE function $\mathrm{mmse}_{M_i}(\gamma)$ versus SNR $\gamma$ for representative $M_i \in \{4, 16, 64, 256\}$. 
    }
    \label{fig:mmse_lut}
\end{figure}

Specifically, let $\{b_i^{(t)}\}$ denote the quantization bit allocation at iteration $t$. We evaluate $\mathrm{SNR}_i^{(t)} = P_i/(\sigma^2 + q_i^{(t)})$ 
with $q_i^{(t)} = c (P_i+\sigma^2) \cdot 2^{-b_i^{(t)}}$ and compute the saturation factor\footnote{Since 
$\mathrm{mmse}_{M_i}(\cdot)$ admits no closed-form expression for finite-alphabet inputs, we precompute it over a dense SNR grid for each modulation order and store the values in a lookup table (Fig.~\ref{fig:mmse_lut}), as commonly done in MWF~\cite{LozanoTulinoVerdu}.}
\begin{align} \label{eq:mercury_update}
  m_i^{(t)} \triangleq 
  \frac{P_i \, \mathrm{mmse}_{M_i}(\mathrm{SNR}_i^{(t)})}{(\sigma^2 + q_i^{(t)})^2} .
\end{align}
Treating $m_i^{(t)}$ as fixed in~\eqref{eq:kkt_fa_exact} reduces the KKT condition to $q_i \cdot m_i^{(t)} = \mu^{(t+1)}$ for $b_i>0$, which yields
\begin{align} \label{eq:bit_update_fa}
b_i^{(t+1)} = 
\max\left\{
0 ,\, 
\log_2{\frac{\lambda_i \cdot m_i^{(t)}}{\mu^{(t+1)}}}
\right\},
\quad
\lambda_i \triangleq c(P_i + \sigma^2).
\end{align}
The multiplier $\mu^{(t+1)}$ is determined by the budget constraint. Let $\CMcal{A}^{(t+1)} = \{i \mid b_i^{(t+1)} > 0\}$ denote the set of active channels at iteration $t+1$. Conditioned on $\CMcal{A}^{(t+1)}$, the multiplier admits the explicit expression
\begin{align} \label{eq:mu_update_fa}
    |\CMcal{A}^{(t+1)}|\log_2 \mu^{(t+1)} 
    = \sum_{i \in \CMcal{A}^{(t+1)}} \log_2 \big( m_i^{(t)} \lambda_i \big) - B.
\end{align}
The complete process of the proposed RMWF is summarized in Algorithm~\ref{alg:rmwf_toy}.

Since the active set and the multiplier are interdependent, they must be determined jointly. 
We resolve this by initializing the active set to the full index set and alternating between \eqref{eq:mu_update_fa} and \eqref{eq:bit_update_fa} until the active set stabilizes. 
For a fixed iteration $t$, because removing inactive indices cannot decrease the water level required to satisfy the bit budget, indices discarded within this inner loop cannot be reactivated. 
Thus, the active set can only shrink, and the inner loop terminates after at most $N$ iterations.
Once $\CMcal{A}^{(t+1)}$ is fixed, each iteration of Algorithm~\ref{alg:rmwf_toy} therefore reduces to two elementary steps: \eqref{eq:mercury_update} refreshing the saturation factor $m_i^{(t)}$ from the current allocation, and the joint~\eqref{eq:bit_update_fa}--\eqref{eq:mu_update_fa} update producing the next allocation in closed form. 
This process repeats until $\{b_i^{(t)}\}$ converges.



\begin{algorithm}[t]
\caption{RMWF for $N$-parallel channel}
\label{alg:rmwf_toy}
\begin{algorithmic}[1]
\STATE \textbf{Input:} Powers $\{P_i\}$, modulations 
$\{M_i\}$, noise variance $\sigma^2$, quantization 
constant $c$, budget $B$, tolerance $\epsilon$
\STATE \textbf{Initialize:} $b_i^{(0)} \leftarrow 
B/N$ for all $i$; $\lambda_i \leftarrow c(P_i + \sigma^2)$; 
$t \leftarrow 0$
\REPEAT
  \STATE Compute $q_i^{(t)} = \lambda_i 
         \cdot 2^{-b_i^{(t)}}$ for all $i$
  \STATE Compute $\mathrm{SNR}_i^{(t)} = P_i / 
         (\sigma^2 + q_i^{(t)})$ for all $i$
  \STATE Update $m_i^{(t)}$ 
         via~\eqref{eq:mercury_update}
  \STATE \textbf{Joint $(\mu, \CMcal{A})$ update:} initialize $\CMcal{A}^{(t+1)} \leftarrow \{1, \dots, N\}$
  \REPEAT
    \STATE Compute $\mu^{(t+1)}$ 
           via~\eqref{eq:mu_update_fa} using $\CMcal{A}^{(t+1)}$
    \STATE Update $b_i^{(t+1)}$ via~\eqref{eq:bit_update_fa}
    \STATE $\CMcal{A}^{(t+1)} \leftarrow \{i \mid b_i^{(t+1)} > 0\}$
  \UNTIL{$\CMcal{A}^{(t+1)}$ unchanged}
  \STATE $t \leftarrow t + 1$
\UNTIL{$\|\mathbf{b}^{(t)} - \mathbf{b}^{(t-1)}\|_\infty 
       < \epsilon$}
\STATE \textbf{Output:} $\{b_i^{(t)}\}$
\end{algorithmic}
\end{algorithm}

We now interpret the structure of the proposed RMWF described in Algorithm~\ref{alg:rmwf_toy}.
At convergence, the iterate $\{b_i^{(t)}\}$ stabilizes at a fixed point $\{b_i^{\star}\}$ with associated saturation factor $m_i^{\star}$ and multiplier $\mu^{\star}$. 
Under $\gamma_i^\star \triangleq P_i / (\sigma^2 + q_i^\star) \gg 1$, 
the converged allocation admits the form\footnote{For notational clarity 
we set $c=1$ in the algebraic decomposition below; the actual algorithm 
retains $c$ as defined.
Reinstating $c$ corresponds to the rescaling 
$\mu^\star \to c\mu^\star$ and does not affect the structural form.}
\begin{align} 
  b_i^{\star} 
  &= \log_2{\left(\frac{P_i + \sigma^2}{\sigma^2+q_i^\star}\right)} + \log_2{\left(\frac{P_i}{\sigma^2+q_i^\star} \cdot \mathrm{mmse}_{M_i}(\gamma_i^\star)
  \right)} \nonumber \\
  &\qquad - \log_2 \mu^{\star} \\
  &\approx 
  \underbrace{\log_2{\gamma_i^\star}}_{\text{vessel height}}
  - \underbrace{\log_2 \mu^{\star}}_{\text{water level}}
  - \underbrace{\log_2{\frac{1}{\gamma_i^\star \cdot \mathrm{mmse}_{M_i}(\gamma_i^\star)}}}_{\text{mercury level}}, \label{eq:rmwf_sol}
\end{align}
where the approximation follows from 
\begin{align}
    \log_2{\left(\gamma_i^\star + \frac{\sigma^2}{\sigma^2 + q_i^\star}\right)} 
    \approx \log_2 \gamma_i^\star,
\end{align}
which holds tightly when $\gamma_i^\star$ is sufficiently large, since $q_i \ge 0$.

The above decomposition reveals three crucial ingredients of our allocation.
The water level $\log_2 \mu^\star$ is a universal threshold set by the total budget; bits are assigned only to the excess vessel height above it~\cite{CoverThomas}. 
This plays the same role as the water level in classical RWF, balancing the bit budget across all channels through a single shared threshold.
The vessel height $\log_2 \gamma_i^\star$ captures the post-quantization effective strength, hence depends on the final allocation through $q_i^\star$: a higher post-quantization SNR yields a taller vessel and more room for bits.
Notably, unlike MWF for power allocation~\cite{LozanoTulinoVerdu}, where the per-channel gain is fixed by the channel, the vessel height here is self-referential: it depends on the allocation $b_i^\star$ itself through $q_i^\star$, a structural feature unique to RMWF bit allocation.
The mercury level $\log_2(1/(\gamma_i^\star \mathrm{mmse}_{M_i}(\gamma_i^\star)))$ tracks the decoder's saturation state: when the constellation is well-resolved, $\mathrm{mmse}_{M_i} \to 0$ and the mercury grows, leaving little room for additional bits; far from saturation it imposes only a mild correction.

For a Gaussian source, $\mathrm{mmse}_{\rm G}(\gamma) = 1/(1+\gamma) \approx 1/\gamma$ at high SNR, so $\gamma\,\mathrm{mmse}_{\rm G}(\gamma) \to 1$ and the mercury vanishes. This recovers $b_i^\star = \log_2(\gamma_i^\star/\mu^\star)$, i.e., an RWF form. 
Under finite-alphabet inputs, the non-vanishing mercury diverts bits from saturated to unsaturated channels through the shared water level. Motivated by this analogy with MWF for power allocation, we term the proposed algorithm RMWF.

\begin{remark}[Connection to classical RWF]
\label{rem:classical_rwf}
\normalfont
Classical RWF~\cite{CoverThomas} solves the sum-distortion minimization for $N$ parallel Gaussian sources $\{x_i\}$ with variances $\{\sigma_{x_i}^2\}$ under a total bit budget, yielding $b_i^\star = \max\{0, \log_2(\sigma_{x_i}^2/\mu)\}$. Our formulation maximizes MI rather than minimizing distortion, yet in the Gaussian high-SNR limit it produces the same structural form $b_i^\star = \log_2(\gamma_i^\star/\mu^\star)$.
\end{remark}

\begin{remark}[A four-way correspondence between power and bit allocation]
\normalfont
Together with MWF for power allocation~\cite{LozanoTulinoVerdu}, RMWF completes a four-way correspondence summarized in Table~\ref{tab:correspondence}: the mercury correction emerges on both sides of the communication chain when the Gaussian input assumption is replaced by finite-alphabet constellations.
The Gaussian baseline hides this saturation, since the mercury level vanishes for Gaussian inputs, and the two problems reduce to the classical WF and classical RWF forms, respectively.
\end{remark}

\begin{remark}[Concavity in the high-rate regime]
\label{rem:convexity_toy}
\normalfont
Problem~\eqref{eq:opt_fa} is non-convex in $\{b_i\}$ in general, but concavity is recovered in the high-rate regime. To see this, a direct calculation gives
\begin{align} \label{eq:snr_second_deriv}
    \frac{d^2 \mathrm{SNR}_i}{db_i^2} 
    = (\ln 2)^2 P_i\, q_i \cdot \frac{q_i - \sigma^2}{(\sigma^2 + q_i)^3},
\end{align}
whose sign is governed by $q_i - \sigma^2$. Hence, when $q_i < \sigma^2\, \,\forall i$, $\mathrm{SNR}_i(b_i)$ is concave in each $b_i$. Combined with the concavity and monotonicity of $I_{M_i}(\gamma)$ in $\gamma$~\cite{GuoShamaiVerdu}, the objective in~\eqref{eq:opt_fa} becomes concave in $\{b_i\}$ \cite{boyd2004convex}. 
Therefore, the stationary point in the same region is a local maximum.
\end{remark}

\begin{figure}[t]     
\centerline{\resizebox{1.0\columnwidth}{!}{\includegraphics{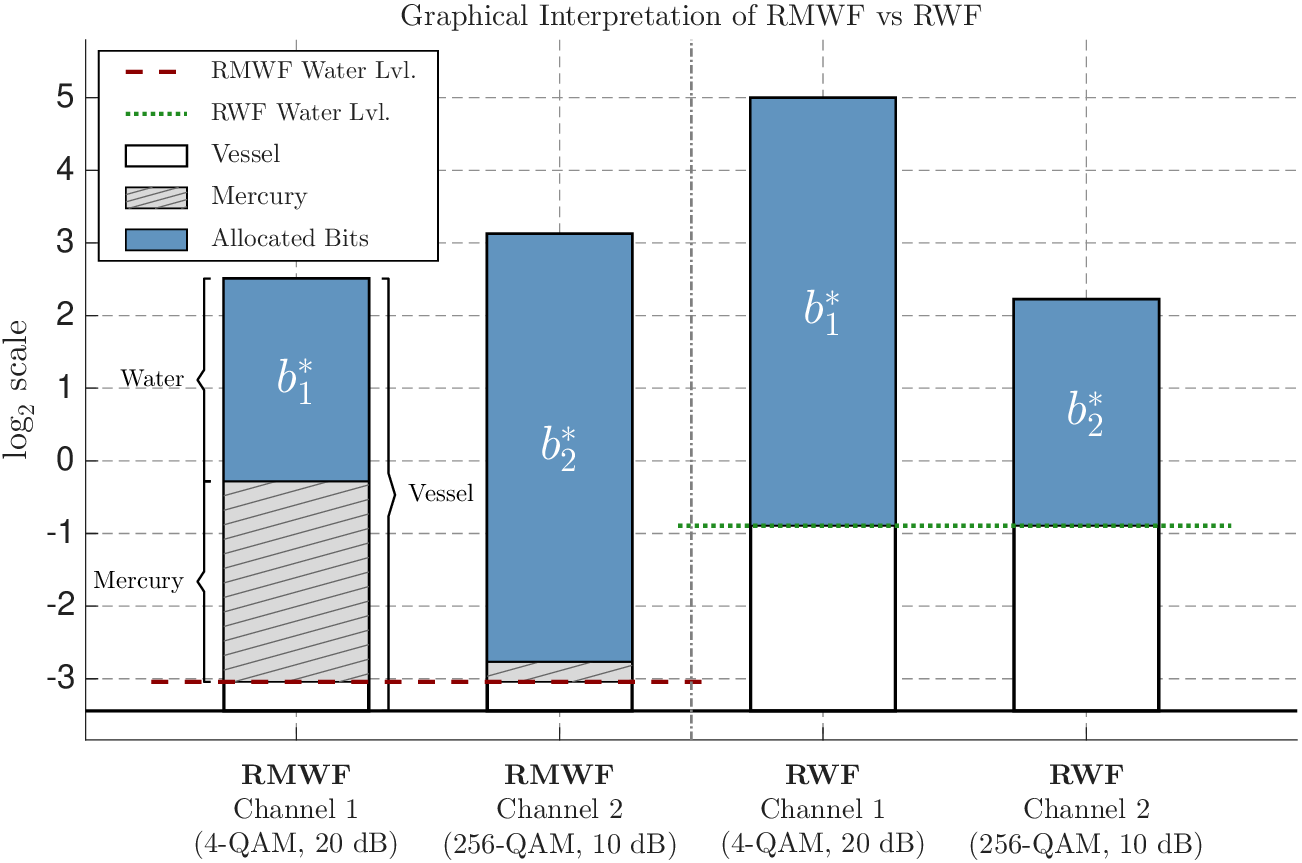}}}
    \caption{Graphical interpretation of RMWF and RWF, based on the high-SNR decomposition in~\eqref{eq:rmwf_sol}. 
    Channel 1 uses 4-QAM with $P_1/\sigma^2 = 20$~dB, and channel 2 uses 256-QAM with $P_2/\sigma^2 = 10$~dB.
    Total quantization bit budget $B=12$.
    The achieved average sum GMI is approximately 5.13 and 4 bits/symbol for RMWF and RWF, respectively, compared to an ideal 5.29 bits/symbol under infinite fronthaul capacity.
    }\label{fig:rmwf_rwf_interp}
\end{figure}

To clearly contrast the RMWF bit allocation incorporating the finite-alphabet inputs against the conventional Gaussian baseline, we illustrate the behaviors of RMWF and RWF in Fig.~\ref{fig:rmwf_rwf_interp}. 
Specifically, we consider a two-channel setup with a total quantization bit budget $B=12$; the first channel is configured with 4-QAM at a high transmit power-to-noise ratio of 20 dB, while the second channel employs 256-QAM at a relatively low 10 dB.
For each channel, the vessel height $\log_2\gamma_i^\star$ encodes the SNR, the mercury level $\log_2(1/(\gamma_i^\star\mathrm{mmse}_{M_i}(\gamma_i^\star)))$ sits above the bottom, and the allocated bits are the clearance between the shared water level $\log_2 \mu^\star$ and the top of the mercury.

The key contrast appears in channel~1. Under RMWF, the high SNR drives $\mathrm{mmse}_4(\gamma_1^\star)\to 0$, producing a thick mercury that blocks additional bit allocation. 
RWF, treating the input as Gaussian, has $\gamma\,\mathrm{mmse}_{\rm G}(\gamma)\to 1$ and therefore zero mercury; it continues pouring bits into the already-saturated channel. 
The vessel heights themselves also differ, because $\gamma_i^\star$ depends on the final allocation: RMWF assigns fewer bits to channel~1 and more bits to channel~2.
For channel 2, the constellation is far from saturation, so the RMWF mercury stays thin and RMWF directs more bits there, whereas RWF, blind to this, does the opposite.
In short, RMWF discriminates by joint channel strength and modulation-dependent saturation, while RWF discriminates by signal power alone and wastes bits on resolved constellations.

\begin{table*}[t]
\centering
\caption{Four-way correspondence between power and bit allocation for MI maximization.}
\label{tab:correspondence}
\renewcommand{\arraystretch}{1.5}
\begin{tabular*}{0.85\textwidth}{@{\extracolsep{\fill}} c|cc @{}}
\toprule
 & Gaussian inputs & Finite-alphabet inputs \\
\midrule
Power allocation & Waterfilling & Mercury/waterfilling~\cite{LozanoTulinoVerdu} \\
 & $p_i = \left(1/\mu - 1/g_i\right)^+$ 
 & $p_i = \mathrm{mmse}_{M_i}^{-1}\left(\min\{1, \mu/g_i\}\right) / g_i$ \\
Bit allocation & Reverse waterfilling~\cite{CoverThomas} & \textbf{Reverse mercury/waterfilling (this work)} \\
 & $b_i = \left(\log_2\left(\lambda_i/\mu\right)\right)^+$ 
 & $b_i = \left(\log_2{\gamma_i} - \log_2{\mu} - \log_2{1/(\gamma_i\cdot \mathrm{mmse}_{M_i}(\gamma_i))}\right)^+$ 
 \\
\bottomrule
\end{tabular*}
\end{table*}

\section{Iterative RMWF-based Fronthaul Compression Framework}
\label{sec:rmwf}

We now extend the RMWF framework developed in the toy example of Section~\ref{sec:toy_fa} to the multi-RU, multi-user C-RAN fronthaul compression problem.

\subsection{Sum Fronthaul Constraint}

{\bf{Lagrangian and KKT Condition:}}
Associating the multiplier $\mu > 0$ with the sum fronthaul budget in~\eqref{eq:problem_sum}, the corresponding Lagrangian is given by 
\begin{align} \label{eq:lagrangian_p1}
  \CMcal{L}\big(\{b_{m,n}\}, \mu\big) 
  = \sum_{k=1}^{K} 
    I_{M_k}\big(\mathrm{SINR}_k\big) 
  - \mu  \left(\sum_{m=1}^{M} \sum_{n=1}^{N} 
                b_{m,n} - C_{\sf tot}\right),
\end{align}
and the KKT necessary condition requires
\begin{align} \label{eq:kkt_p1}
  \sum_{k=1}^{K} 
  \frac{d I_{M_k}}{d \mathrm{SINR}_k} 
  \cdot \frac{\partial \mathrm{SINR}_k}
              {\partial b_{m,n}} 
  \begin{cases}
       = \mu & \text{if } b_{m,n} > 0 \\
       \leq \mu & \text{if } b_{m,n} = 0
  \end{cases}\,,
  \quad \forall (m, n).
\end{align}
The first factor in~\eqref{eq:kkt_p1} is obtained by the I-MMSE relation:
\begin{align} \label{eq:immse_user}
  \frac{d I_{M_k}}{d \mathrm{SINR}_k} 
  = \frac{1}{\ln 2} \, 
    \mathrm{mmse}_{M_k}\big(\mathrm{SINR}_k\big).
\end{align}
The second factor requires differentiating the 
post-LMMSE SINR through its dependence on the 
quantization noise covariance $\mathbf{Q}$. 
The following lemma provides the desired expression.

\begin{lemma}[Derivative of post-LMMSE SINR]
\label{lem:dsinr}
Let $\mathrm{SINR}_k = P_k \mathbf{h}_k^{\sf H} 
\mathbf{R}_k^{-1} \mathbf{h}_k$, where $\mathbf{R}_k \triangleq \sum_{j \neq k} P_j \mathbf{h}_j \mathbf{h}_j^{\sf H} + \sigma^2 \mathbf{I}_{MN} + \mathbf{Q}$. 
Then,
\begin{align} \label{eq:dsinr}
  \frac{\partial \mathrm{SINR}_k}{\partial b_{m,n}} 
  = (\ln 2) \, q_{m,n} \cdot P_k \, 
    \big|\mathbf{h}_k^{\sf H} \mathbf{R}_k^{-1} 
         \mathbf{\tilde u}_{m,n}\big|^2,
\end{align}
where $\mathbf{\tilde u}_{m,n} \triangleq \mathbf{v}_m \otimes \mathbf{u}_{m,n} \in \mathbb{C}^{MN}$, with $\mathbf{u}_{m,n} \in \mathbb{C}^N$ being the $n$-th eigenvector of the local covariance matrix $\mathbf{R}_{\mathbf{y}_m}$ at RU $m$, and $\mathbf{v}_m \in \mathbb{R}^M$ denoting the $m$-th standard basis vector.
\end{lemma}

\begin{proof}
The dependence of $\mathrm{SINR}_k$ on $b_{m,n}$ is 
mediated entirely through the quantization noise 
covariance $\mathbf{Q}$. By the chain rule,
\begin{align}
  \frac{\partial \mathrm{SINR}_k}{\partial b_{m,n}} 
  = P_k \mathbf{h}_k^{\sf H} 
    \frac{\partial \mathbf{R}_k^{-1}}
         {\partial b_{m,n}} \mathbf{h}_k.
\end{align}
The matrix derivative identity 
$\partial \mathbf{R}_k^{-1}/\partial b_{m,n} 
= -\mathbf{R}_k^{-1} (\partial \mathbf{R}_k/
\partial b_{m,n}) \mathbf{R}_k^{-1}$ together with 
$\partial \mathbf{R}_k/\partial b_{m,n} 
= \partial \mathbf{Q}/\partial b_{m,n} 
= -(\ln 2) q_{m,n} \mathbf{\tilde u}_{m,n} \mathbf{\tilde u}_{m,n}^{\sf H}$, where the last equality follows from $q_{m,n} = c \lambda_{m,n} \cdot 2^{-b_{m,n}}$ and the block structure of $\mathbf{Q}$ gives
\begin{align}
  \frac{\partial \mathbf{R}_k^{-1}}{\partial b_{m,n}} 
  = (\ln 2) q_{m,n} \mathbf{R}_k^{-1} 
    \mathbf{\tilde u}_{m,n} \mathbf{\tilde u}_{m,n}^{\sf H} 
    \mathbf{R}_k^{-1}.
\end{align}
Substituting this into the chain rule expression 
yields~\eqref{eq:dsinr}.
\end{proof}

Substituting~\eqref{eq:immse_user} 
and~\eqref{eq:dsinr} into~\eqref{eq:kkt_p1} yields the 
exact KKT condition
\begin{align} \label{eq:kkt_p1_exact}
  q_{m,n} \cdot S_{m,n}\big(\{b_{m,n}\}\big) 
  \begin{cases}
       = \mu & \text{if } b_{m,n} > 0 \\
       \leq \mu & \text{if } b_{m,n} = 0
  \end{cases}\,,
  \quad \forall (m, n),
\end{align}
where $S_{m,n}\big(\{b_{m,n}\}\big) $ is the user-aggregated saturation score defined as
\begin{align} \label{eq:score_def}
  S_{m,n}\big(\{b_{m,n}\}\big) 
  \triangleq \sum_{k=1}^{K} 
    P_k \, \mathrm{mmse}_{M_k}\big(\mathrm{SINR}_k\big) 
    \cdot \big|\mathbf{h}_k^{\sf H} 
                \mathbf{R}_k^{-1} 
                \mathbf{\tilde u}_{m,n}\big|^2.
\end{align}
The score $S_{m,n}$ generalizes the toy-example mercury weight $m_i$~\eqref{eq:mercury_update}
to the multi-user, multi-RU C-RAN setting.
The $(m,n)$-th coefficient contributes to each user $k$ through the product of two factors: the LMMSE projection $|\mathbf{h}_k^{\sf H}\mathbf{R}_k^{-1}\mathbf{\tilde u}_{m,n}|^2$, measuring how strongly user $k$'s detection leverages it, and $P_k\,\mathrm{mmse}_{M_k}(\mathrm{SINR}_k)$, measuring how much a marginal SINR gain is still worth under user $k$'s saturation state.
Their product is the coefficient's marginal contribution to user $k$'s GMI, so $S_{m,n}$ scores the total contribution of the $(m,n)$-th coefficient to the sum GMI.

{\bf{Iterative RMWF Algorithm:}}
The KKT condition~\eqref{eq:kkt_p1_exact} couples 
the bit allocations $\{b_{m,n}\}$ through the score 
$S_{m,n}$, which itself depends on $\{b_{m,n}\}$ via 
the post-LMMSE SINRs $\{\mathrm{SINR}_k\}$ (through 
$\mathbf{R}_k$). 
Specifically, given the iterate $\{b_{m,n}^{(t)}\}$ at step $t$, we compute the post-quantization SINRs $\{\mathrm{SINR}_k^{(t)}\}$ via~\eqref{eq:sinr} with $\mathbf{Q}^{(t)}$ assembled from $q_{m,n}^{(t)} = c \lambda_{m,n} \cdot 2^{-b_{m,n}^{(t)}}$, and refresh the scores
\begin{align} \label{eq:score_update}
  S_{m,n}^{(t)} 
  \triangleq \sum_{k=1}^{K} 
    P_k \, \mathrm{mmse}_{M_k}\big(\mathrm{SINR}_k^{(t)}\big) 
    \cdot \big|\mathbf{h}_k^{\sf H} 
              \big(\mathbf{R}_k^{(t)}\big)^{-1} 
              \mathbf{\tilde u}_{m,n}\big|^2.
\end{align}
By treating $\{S_{m,n}^{(t)}\}$ as fixed in~\eqref{eq:kkt_p1_exact}, the KKT condition reduces to
\begin{align} \label{eq:bit_update_p1}
  b_{m,n}^{(t+1)} 
  = \max\left\{
  0,\,
  \log_2 \frac{c\lambda_{m,n} \cdot S_{m,n}^{(t)}}
                {\mu^{(t+1)}}\right\},
\end{align}
where we have substituted $q_{m,n} = c \lambda_{m,n} \cdot 2^{-b_{m,n}}$ and 
written $\lambda_{m,n}$ for the KLT eigenvalue. 


The multiplier $\mu^{(t+1)}$ is determined by 
imposing the sum-rate budget 
$\sum_{m,n} b_{m,n}^{(t+1)} = C_{\sf tot}$ 
on~\eqref{eq:bit_update_p1}. Let 
$\CMcal{A}^{(t+1)} = \{(m, n) \mid b_{m,n}^{(t+1)} > 0\}$ 
denote the set of active KLT coefficients at 
iteration $t+1$. Conditioned on 
$\CMcal{A}^{(t+1)}$, the multiplier admits the 
explicit expression
\begin{align} \label{eq:mu_closed}
  |\CMcal{A}^{(t+1)}|\log_2 \mu^{(t+1)} 
  = \sum_{(m,n) \in \CMcal{A}^{(t+1)}} 
    \log_2 \big(c\lambda_{m,n} \cdot S_{m,n}^{(t)}\big) 
  - C_{\sf tot}.
\end{align}
Since $\CMcal{A}^{(t+1)}$ itself depends on $\mu^{(t+1)}$ through~\eqref{eq:bit_update_p1}, the active set and the multiplier must be determined jointly. 
We do this by initializing $\CMcal{A}^{(t+1)}$ to the full index set $\{(m, n)\}_{m,n}$ and alternating  between~\eqref{eq:mu_closed} and~\eqref{eq:bit_update_p1} until $\CMcal{A}^{(t+1)}$ stabilizes. 
As in the toy example, this inner loop terminates finitely and no bisection on $\mu$ is required. 
Each iteration thus consists of two elementary steps: \eqref{eq:score_update} refreshing the scores, and~\eqref{eq:bit_update_p1}--\eqref{eq:mu_closed} updating the allocation in closed form.
This process repeats until convergence. The full procedure is summarized in Algorithm~\ref{alg:rmwf_p1}.

The per-iteration complexity is dominated by the $K$ matrix inverses 
$(\mathbf{R}_k^{(t)})^{-1} \in \mathbb{C}^{MN \times MN}$ in the SINR step, 
costing $\mathcal{O}(K(MN)^3)$. Empirically, the iteration converges in 
$T = 5$--$8$ steps for typical C-RAN configurations 
(see Section~\ref{subsec:convergence}), making the overall complexity 
$\mathcal{O}(T K (MN)^3)$. 


\begin{algorithm}[t]
\caption{Iterative RMWF for (P1)}
\label{alg:rmwf_p1}
\begin{algorithmic}[1]
\STATE \textbf{Input:} Channels $\{\mathbf{H}_m\}$, 
powers $\{P_k\}$, modulations $\{M_k\}$, 
eigenvalues $\{\lambda_{m,n}\}$, eigenvectors 
$\{\mathbf{\tilde u}_{m,n}\}$, budget $C_{\sf tot}$, 
tolerance $\epsilon$
\STATE \textbf{Initialize:} $b_{m,n}^{(0)} 
\leftarrow C_{\sf tot}/(MN)$ for all $(m, n)$; 
$t \leftarrow 0$
\REPEAT
  \STATE Form $\mathbf{Q}^{(t)} = 
         \mathrm{blkdiag}(\mathbf{U}_m \mathbf{Q}_m^{(t)} 
         \mathbf{U}_m^{\sf H})$ with 
         $q_{m,n}^{(t)} = c \lambda_{m,n} \cdot 
         2^{-b_{m,n}^{(t)}}$
  \STATE Compute $\mathrm{SINR}_k^{(t)} 
         = P_k \mathbf{h}_k^{\sf H} 
         (\mathbf{R}_k^{(t)})^{-1} \mathbf{h}_k$ 
         for all $k$
  \STATE Update scores $S_{m,n}^{(t)}$ 
         via~\eqref{eq:score_update}
  \STATE \textbf{Joint $(\mu, \CMcal{A})$ update:} 
         initialize $\CMcal{A}^{(t+1)} \leftarrow 
         \{(m, n) \mid 1 \le m \le M, 1 \le n \le N\}$
  \REPEAT
    \STATE Compute $\mu^{(t+1)}$ 
           via~\eqref{eq:mu_closed} using 
           $\CMcal{A}^{(t+1)}$
    \STATE Update $b_{m,n}^{(t+1)}$ 
           via~\eqref{eq:bit_update_p1}
    \STATE $\CMcal{A}^{(t+1)} \leftarrow 
           \{(m, n) \mid b_{m,n}^{(t+1)} > 0\}$
  \UNTIL{$\CMcal{A}^{(t+1)}$ unchanged}
  \STATE $t \leftarrow t + 1$
\UNTIL{$\|\mathbf{b}^{(t)} - 
       \mathbf{b}^{(t-1)}\|_{\infty} < \epsilon$}
\STATE \textbf{Output:} $\{b_{m,n}^{(t)}\}$
\end{algorithmic}
\end{algorithm}


\subsection{Per-RU Fronthaul Constraint}

Problem~(P2) inherits the KKT structure and fixed-point iteration of
Section~\ref{sec:rmwf}-A verbatim; only the global multiplier $\mu$
splits into $M$ RU-local multipliers $\{\mu_m\}$. Associating $\mu_m > 0$ with the $m$-th constraint in~\eqref{eq:problem_per_ru} yields the
localized KKT condition
\begin{align} \label{eq:kkt_p2_exact}
  q_{m,n} \cdot S_{m,n}\big(\{b_{m,n}\}\big)
  \begin{cases}
       = \mu_m & \text{if } b_{m,n} > 0 \\
       \leq \mu_m & \text{if } b_{m,n} = 0
  \end{cases},
  \quad \forall (m, n),
\end{align}
with $S_{m,n}$ defined as in~\eqref{eq:score_def}. The bits are then
updated by
\begin{align} \label{eq:bit_update_p2}
  b_{m,n}^{(t+1)} = \max \left\{0,\,
    \log_2 \frac{c\lambda_{m,n} \cdot S_{m,n}^{(t)}}
                {\mu_m^{(t+1)}}\right\},
\end{align}
where each $\mu_m^{(t+1)}$ is determined by imposing the local budget
$\sum_n b_{m,n}^{(t+1)} = C_m$, yielding the closed form
\begin{align} \label{eq:mu_closed_p2}
  |\CMcal{A}_m^{(t+1)}| \log_2 \mu_m^{(t+1)}
  = \sum_{n \in \CMcal{A}_m^{(t+1)}}
    \log_2 \big(c\lambda_{m,n} \cdot S_{m,n}^{(t)}\big) - C_m,
\end{align}
with $\CMcal{A}_m^{(t+1)}=\{n\mid b_{m,n}^{(t+1)}>0\}$. As in~(P1),
each $(\mu_m, \CMcal{A}_m)$ pair is determined jointly by alternation;
the inner loop now runs independently per RU. Structurally,~(P2) decomposes into $M$ local RMWF problems coupled indirectly through the scores: a bit-starved user inflates the scores of all coefficients serving it across all RUs, prompting each RU to reallocate its local budget toward that user. The full procedure appears in Algorithm~\ref{alg:rmwf_p2}, with overall complexity $\mathcal{O}(TK(MN)^3)$ identical to~(P1).



\begin{remark}[Concavity of (P1) and (P2)]
\label{rem:concavity_cran}
\normalfont
Problems (P1) and (P2) are non-convex in $\{b_{m,n}\}$ in general, 
but concavity is recovered in the high-rate regime $q_{m,n} < \sigma^2\,\, \forall m,n$. 
To establish $\nabla^2 \mathrm{SINR}_k \preceq \mathbf{0}$, 
$\mathbf{R}_k^{-1} \preceq (\sigma^2 \mathbf{I}_{MN} + \mathbf{Q})^{-1}$ 
together with the KLT-diagonal structure of $\mathbf{Q}$ reduces the 
second directional derivative of $\mathrm{SINR}_k$ to the 
per-coefficient comparison 
$2q_{m,n}^2/(\sigma^2+q_{m,n}) \le q_{m,n}$, equivalent to 
$q_{m,n} \le \sigma^2$. 
The subsequent composition with the concavity and monotonicity of $I_{M_k}$ then proceeds identically to Remark~\ref{rem:convexity_toy}, so the stationary point in the same region is a local maximum.
\end{remark}

\begin{algorithm}[t]
\caption{Iterative RMWF for (P2)}
\label{alg:rmwf_p2}
\begin{algorithmic}[1]
\STATE \textbf{Input:} Channels $\{\mathbf{H}_m\}$, 
powers $\{P_k\}$, modulations $\{M_k\}$, 
eigenvalues $\{\lambda_{m,n}\}$, eigenvectors 
$\{\mathbf{\tilde u}_{m,n}\}$, per-RU budgets $\{C_m\}$, 
tolerance $\epsilon$
\STATE \textbf{Initialize:} $b_{m,n}^{(0)} 
\leftarrow C_m/N$ for all $(m, n)$; $t \leftarrow 0$
\REPEAT
  \STATE Form $\mathbf{Q}^{(t)} = 
         \mathrm{blkdiag}(\mathbf{U}_m \mathbf{Q}_m^{(t)} 
         \mathbf{U}_m^{\sf H})$ with 
         $q_{m,n}^{(t)} = c \lambda_{m,n} \cdot 
         2^{-b_{m,n}^{(t)}}$
  \STATE Compute $\mathrm{SINR}_k^{(t)}$ 
         and update scores $S_{m,n}^{(t)}$ 
         via~\eqref{eq:score_update}
  \FOR{$m = 1, \dots, M$}
    \STATE \textbf{Joint $(\mu_m, \CMcal{A}_m)$ update:} 
           initialize $\CMcal{A}_m^{(t+1)} \leftarrow 
           \{1, \dots, N\}$
    \REPEAT
      \STATE Compute $\mu_m^{(t+1)}$ 
             via~\eqref{eq:mu_closed_p2} using 
             $\CMcal{A}_m^{(t+1)}$
      \STATE Update $b_{m,n}^{(t+1)}$ 
             via~\eqref{eq:bit_update_p2} for all $n$
      \STATE $\CMcal{A}_m^{(t+1)} \leftarrow 
             \{n \mid b_{m,n}^{(t+1)} > 0\}$
    \UNTIL{$\CMcal{A}_m^{(t+1)}$ unchanged}
  \ENDFOR
  \STATE $t \leftarrow t + 1$
\UNTIL{$\|\mathbf{b}^{(t)} - 
       \mathbf{b}^{(t-1)}\|_{\infty} < \epsilon$}
\STATE \textbf{Output:} $\{b_{m,n}^{(t)}\}$
\end{algorithmic}
\end{algorithm}


\section{Numerical Results}\label{sec:numerical}


We consider an $M = 7$ cell hexagonal layout with an inter-site distance of $d = 300$~m, where each cell hosts a single RU equipped with $N$ antennas at a height of $30$~m.
In each channel realization, $K = 14$ users are dropped independently: each user is associated with a cell chosen uniformly at random and placed uniformly within that hexagonal cell at a height of $1.5$~m and at least $10$~m from its RU.
Wireless channels are generated using the 3GPP 3D UMi standard channel model at a carrier frequency of $3.5$~GHz with a subcarrier spacing of $30$~kHz, including both pathloss and log-normal shadow fading. 
Each user is subject to a sum-power constraint of $10$~dBm over a physical resource block ($12$ subcarriers $\times$ $14$ OFDM symbols $=168$ resource elements (REs)), which is met with equal power per RE and thus gives $P_k = -12.25$~dBm. 
The per-RE noise power is $-129.2$~dBm, based on a $-174$~dBm/Hz density and a $30$~kHz subcarrier spacing.
Since the large-scale parameters are common across the block, we evaluate one RE per channel realization.
The modulation order $M_k$ for each user is drawn independently and uniformly at random from $\{4, 16, 64, 256\}$ in every channel realization. 
The antenna count $N$ and the fronthaul budget ($C_{\sf tot}$ or $C_m$) vary per figure and are stated in the corresponding caption.



We summarize the methods compared in our experiments as follows.
\begin{itemize}
    \item \textbf{Upper Bound}: The unquantized performance, obtained by setting the quantization error variance $q_{m,n} = 0$ for all $(m,n)$, equivalently $b_{m,n} \to \infty$. 
    This corresponds to a fronthaul with infinite capacity and provides a method-agnostic ceiling on the achievable sum GMI under the same channel and modulation realizations.
    

    \item \textbf{RMWF}:
    The proposed RMWF allocator with explicit finite-alphabet awareness. The MMSE lookup table (Fig.~\ref{fig:mmse_lut}) is constructed for $M_k \in \{4, 16, 64, 256\}$, with piecewise-linear interpolation between adjacent grid points at runtime. 

    
    \item \textbf{RWF}:
    The classical RWF approach derived under the sum-distortion minimization objective with the assumption of Gaussian sources~\cite{CoverThomas}. Unlike the proposed RMWF, this baseline allocates bits by minimizing the quantization error treating the inputs as Gaussian.

    \item \textbf{Equal-Bit}:
    A modulation-agnostic baseline that retains only the $N/2$ dominant KLT coefficients per RU, assigning $2 \lfloor C_{\sf{tot}} / (MN) \rfloor$ bits to each retained coefficient and allocating any remaining bits in 2 bit increments to the strongest ones to exactly satisfy the budget.
\end{itemize}


A scalar mid-rise uniform quantizer is applied at sample level, separately to the in-phase and quadrature components of each KLT coefficient, with subtractive dithering wrapped around it as in Section~\ref{subsec:cran_fa}: for $b_{m,n}/2$ bits per real dimension, the clipping interval $[-A, A]$ with $A = 3\sigma_{\sf r}$ and $\sigma_{\sf r}^2 = \lambda_{m,n}/2$ is partitioned into $2^{b_{m,n}/2}$ uniform bins and the sample is mapped to the corresponding bin-center reconstruction level, with the zero-bit case mapping to zero output. 
The $3\sigma_{\sf r}$ loading (i.e., $\kappa = 3$) yields the granular-regime constant $c = \kappa^2/3 = 3$ in $q_{m,n} = c\,\lambda_{m,n}\,2^{-b_{m,n}}$.
Since bits are assigned per real dimension, the real-valued RMWF/RWF allocations are rounded to non-negative even integers via the largest-remainder method, with any residual budget mismatch absorbed by $\pm 2$-bit adjustments to the largest entries to preserve the exact budget; this rounding granularity also informs the stopping tolerance used in Section~\ref{subsec:convergence}.

\subsection{GMI versus CCMI}\label{subsec:ccmi}


Fig.~\ref{fig:ccmi} validates the use of the CCMI~\eqref{eq:mi_finite_def} 
as the RMWF design objective by directly comparing it with the actual 
achievable rate, the GMI~\eqref{eq:gmi}, for both (P1) and (P2).

After RMWF returns the bit allocation $\{b_{m,n}^\star\}$, the CU recomputes the quantization noise covariance $\mathbf{Q}^\star$ at this fixed allocation and redesigns the LMMSE receiver $\mathbf{W}_{\sf LMMSE}$ in~\eqref{eq:lmmse} accordingly. 
This two-stage use of $\mathbf{Q}$ is by design: the allocation stage requires the differentiable closed form $q_{m,n} = c\,\lambda_{m,n}\,2^{-b_{m,n}}$ to derive the fixed-point update, whereas the receiver stage, with $\{b_{m,n}^\star\}$ already fixed, can adopt the noise statistics, including the clipping contribution beyond the granular term. 
The CU evaluates these statistics from the signal statistics already available for detection, together with the loading factor $\kappa$, the bits it has assigned, and the locally generated dither, requiring no side information.

The GMI and CCMI curves are visually indistinguishable for both constraints across the entire range of $C_{\sf tot}$, so maximizing the CCMI~\eqref{eq:mi_finite_def} is empirically equivalent to maximizing the actual GMI~\eqref{eq:gmi} in our operating regime.
Guided by this validation, the remainder of this section adopts the 
following convention: RMWF performs bit allocation by maximizing 
the CCMI~\eqref{eq:mi_finite_def}, while all reported rates are 
evaluated as the actual GMI~\eqref{eq:gmi} obtained from Monte 
Carlo simulation of the true residual statistics.

\begin{figure}[t]     
\centerline{\resizebox{0.8\columnwidth}{!}{\includegraphics{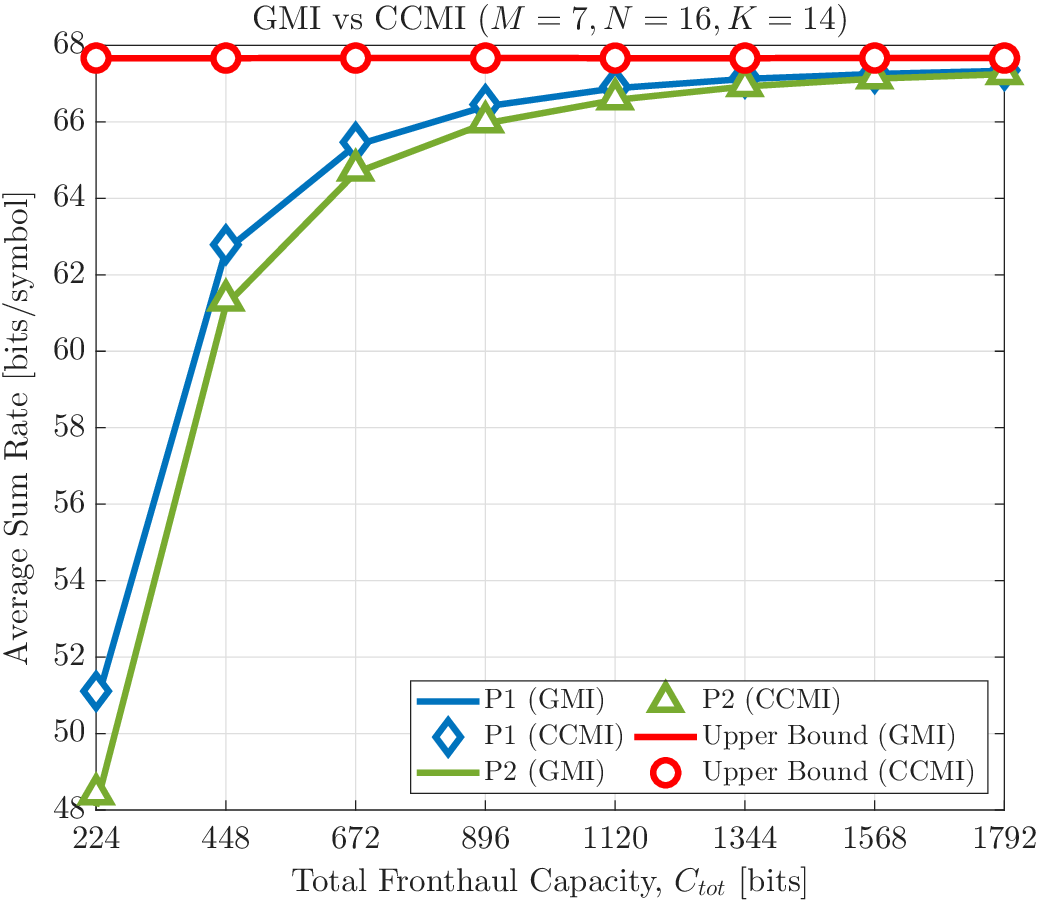}}}
\caption{CCMI-based bit allocation under RMWF: the optimization objective (CCMI, markers) closely tracks the achievable rate (GMI, solid lines), validating the design choice.
$N = 16$ antennas per RU, with $C_m = C_{\sf tot}/M$, $\forall m$, for (P2).}
\label{fig:ccmi}
\end{figure}

\subsection{Average Sum GMI versus Bit Budget}

\begin{figure}[t]
    \centering
    \subfloat[\label{fig:8vs16_P1}]{
        \resizebox{0.8\columnwidth}{!}{\includegraphics{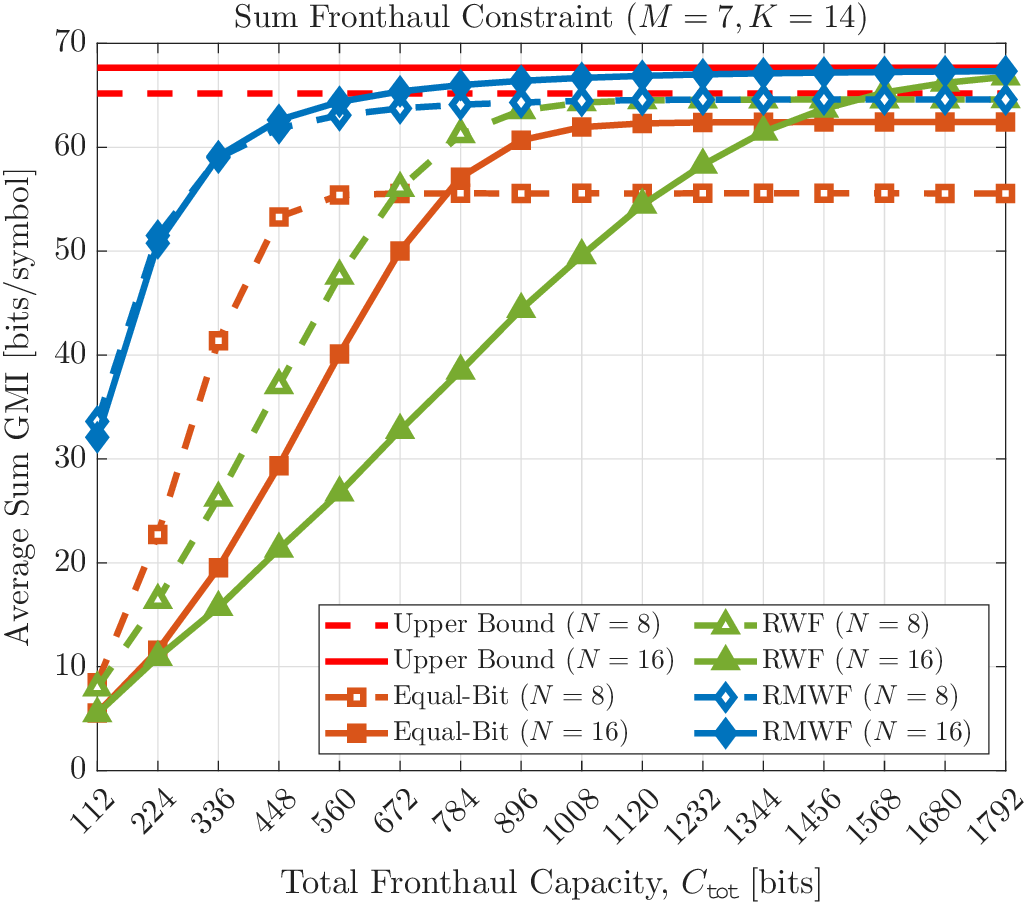}}
    }
    \\ 
    \subfloat[\label{fig:8vs16_P2}]{
        \resizebox{0.8\columnwidth}{!}{\includegraphics{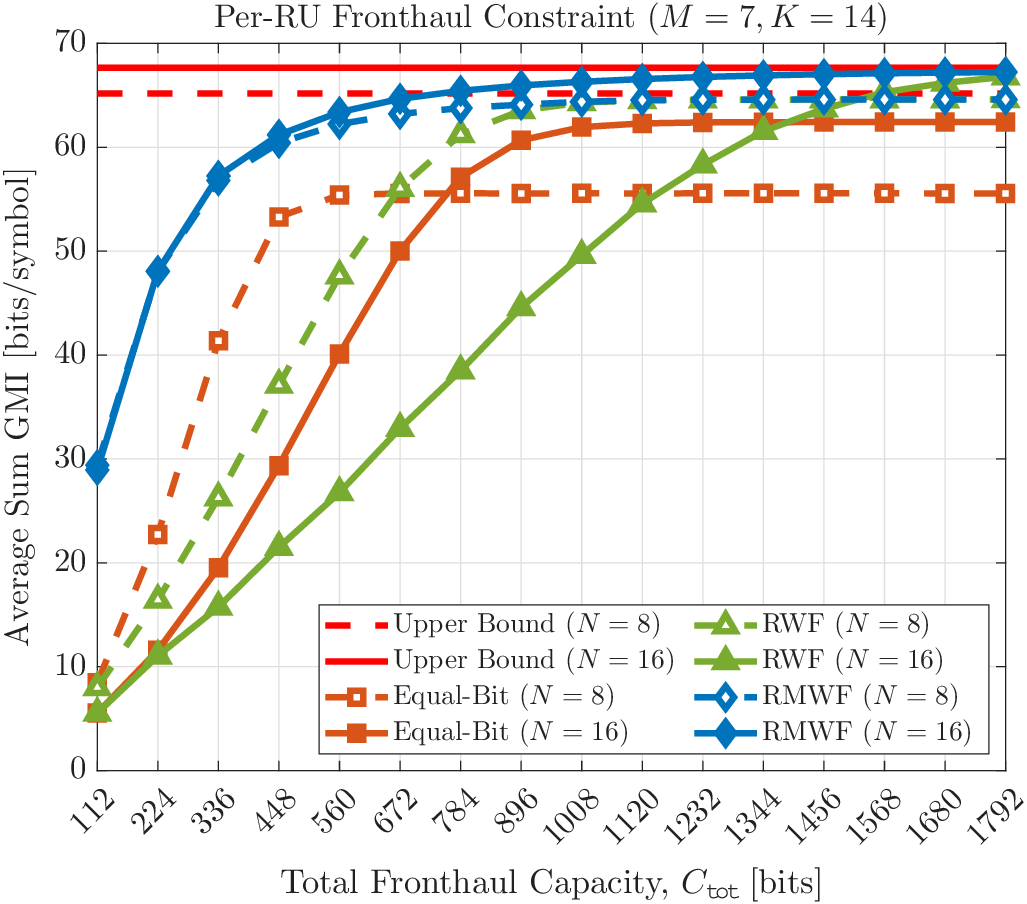}}
    }
    
\caption{Average sum GMI versus total fronthaul capacity $C_{\sf tot}$ for
$N \in \{8, 16\}$ antennas per RU. (a) Sum fronthaul constraint. (b) Per-RU
fronthaul constraint with $C_m = C_{\sf tot}/M$, $\forall m$. }
    \label{fig:8vs16}
\end{figure}

Fig.~\ref{fig:8vs16} plots the average sum GMI as a function of $C_{\sf tot}$ for $N \in \{8, 16\}$ antennas per RU, with Fig.~\ref{fig:8vs16_P1} corresponding to the sum fronthaul constraint and Fig.~\ref{fig:8vs16_P2} to the per-RU constraint. 
Across both constraints and the entire range of $C_{\sf tot}$, the proposed RMWF outperforms both RWF and Equal-Bit. 
At the tightest budget $C_{\sf tot} = 112$ with $N = 16$, RMWF attains a maximum gain of approximately $484.84\%$ ($479.74\%$) over RWF (Equal-Bit) under the sum constraint and $422.84\%$ ($422.84\%$) under the per-RU constraint.

To trace the source of these gains, we first examine how the two baselines turn the bit budget into GMI.
Equal-Bit and RWF share the same design principle: both rank the KLT
coefficients by variance and remain blind to the constellation, differing only in that RWF activates coefficients adaptively through a water level, whereas Equal-Bit rigidly splits the budget over the top half of the eigenvalues. 
This yields a budget-dependent crossover in Fig.~\ref{fig:8vs16}: Equal-Bit outperforms RWF at tight budgets by concentrating bits on the strongest half, but is overtaken beyond $C_{\sf tot} = 560$ and $C_{\sf tot} = 1344$ for $N = 8$ and $N = 16$ respectively, once its excluded coefficients become the bottleneck while RWF keeps activating new ones.
Fixing the active set a priori is thus limiting; yet even RWF's adaptive selection still trails RMWF across the entire range.

RMWF overcomes both limitations simultaneously.
Because it allocates bits to directly maximize the finite-alphabet GMI, its selection is governed jointly by coefficient strength and the modulation-dependent saturation.
As a result, it neither fixes the active set in advance like Equal-Bit, nor, like RWF, wastes bits on coefficients whose marginal contribution to the GMI has already saturated, but instead flexibly selects, in every budget regime, the coefficients that contribute most to the GMI. 
This GMI-driven, finite-alphabet-aware selection is precisely what sustains
RMWF's advantage under tight fronthaul budgets.


\subsection{Average Sum GMI versus the Number of Antennas}\label{subsec:gmi_vs_N}

\begin{figure}[t]
    \centering
    \subfloat[\label{fig:sweep_p1}]{
        \resizebox{0.8\columnwidth}{!}{\includegraphics{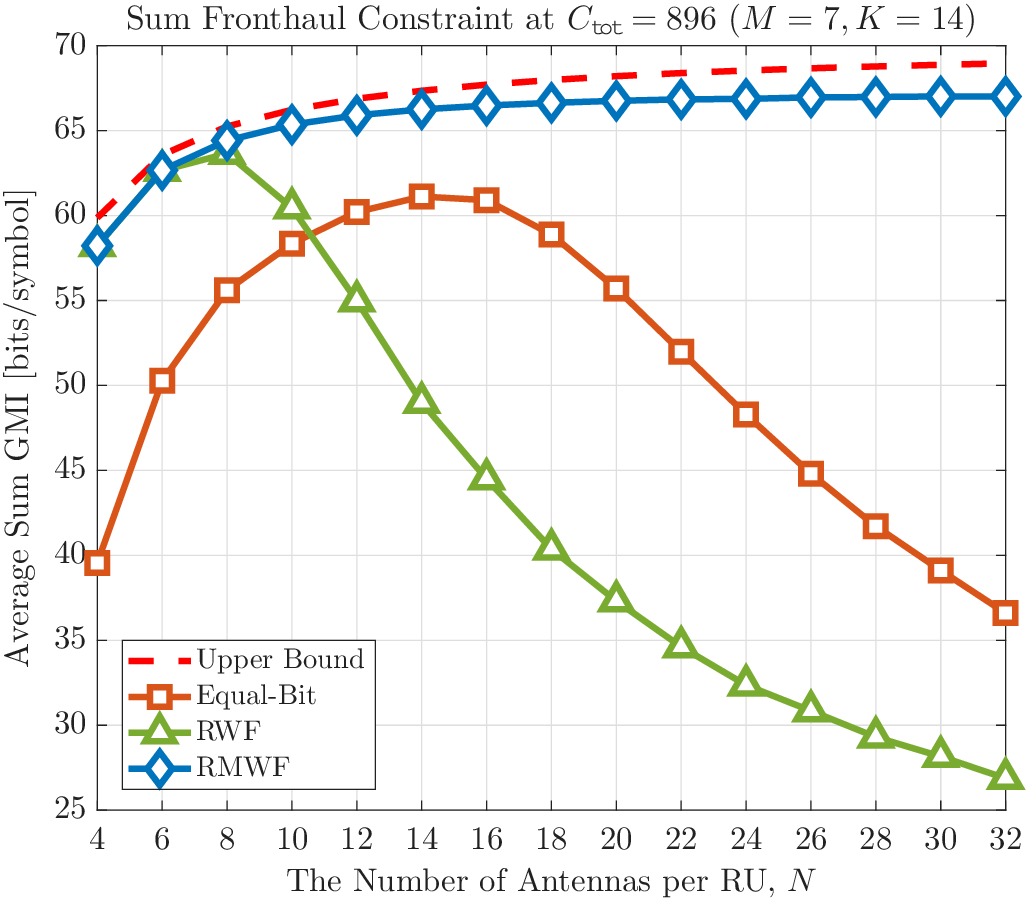}}
    }
    \\ 
    \subfloat[\label{fig:sweep_p2}]{
        \resizebox{0.8\columnwidth}{!}{\includegraphics{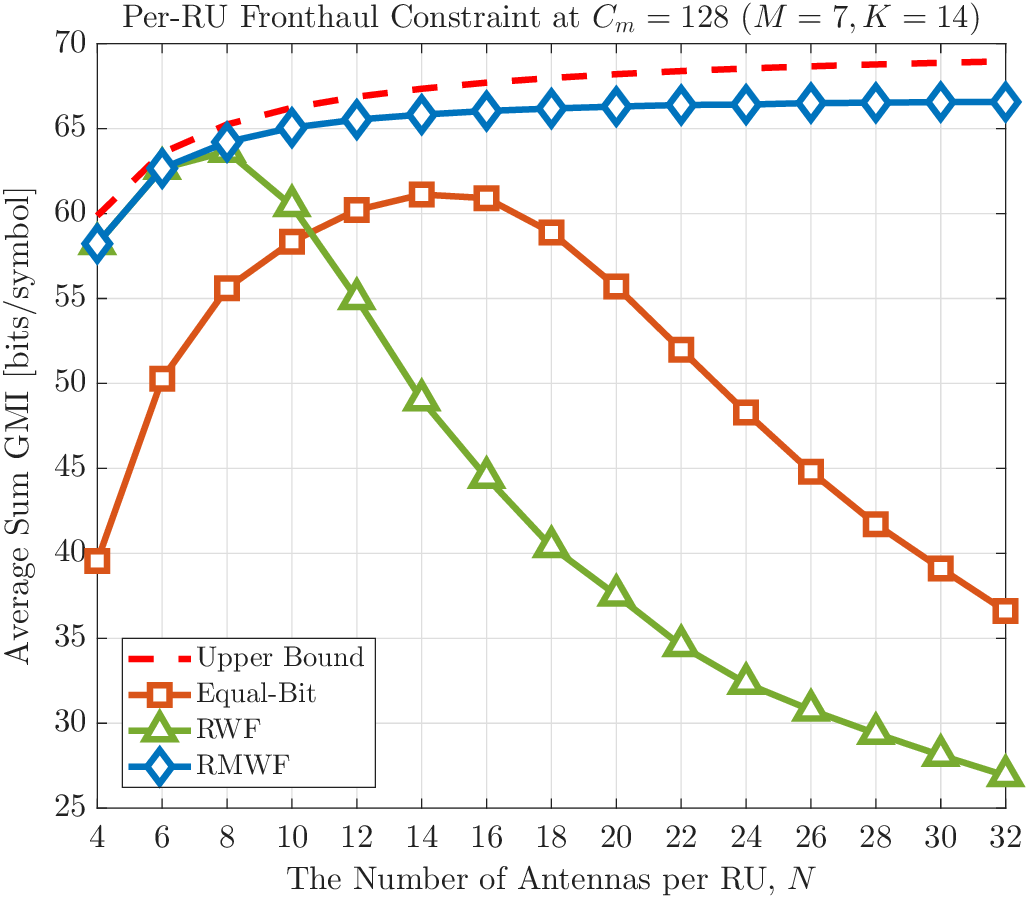}}
    }
\caption{Average sum GMI versus the number of antennas per RU $N$ at a fixed total fronthaul capacity $C_{\sf tot} = 896$. 
(a) Sum fronthaul constraint. (b) Per-RU fronthaul constraint with $C_m = C_{\sf tot}/M = 128$.}
    \label{fig:sweep}
\end{figure}

Fig.~\ref{fig:sweep} plots the average sum GMI as a function of the number of antennas per RU $N$ at a fixed total fronthaul capacity $C_{\sf tot} = 896$, with Fig.~\ref{fig:sweep_p1} and Fig.~\ref{fig:sweep_p2} corresponding to the sum and per-RU fronthaul constraints, respectively. 
The upper bound, which is independent of any allocation strategy, increases monotonically with $N$, reflecting the additional spatial degrees of freedom available at the CU.



The budget $C_{\sf tot}=896$ is a deliberate reference point. 
As Fig.~\ref{fig:8vs16} shows, it provides exactly enough bits for both RMWF and RWF to reach the upper bound at $N=8$.
However, increasing $N$ spreads this fixed budget across more coefficients ($MN=7N$), making the allocation critical. Fig.~\ref{fig:sweep} fixes $C_{\sf tot}=896$ to explicitly capture this transition. 
It reveals a clear shift from a budget-rich regime ($N\le8$) where both schemes track the bound, to a budget-limited regime ($N>8$) where the schemes diverge.

The proposed RMWF settles into a near-flat plateau without any subsequent degradation once the budget-limited regime is entered, after gaining approximately $7.68$ and $7.31$~bits/symbol from $N = 4$ to $N = 12$ under the sum and per-RU constraints, respectively. 
The activation rule in \eqref{eq:bit_update_p1}--\eqref{eq:mu_closed} (analogously \eqref{eq:bit_update_p2}--\eqref{eq:mu_closed_p2} for the per-RU case) explains why: only coefficients with $c\lambda_{m,n} \cdot S_{m,n} > \mu$ enter the active set $\CMcal{A}$, so the extra coefficients introduced by larger $N$, carrying small eigenvalues, simply fall below the water level and are left inactive.
RMWF thus keeps concentrating the budget on the same dominant coefficients regardless of $N$, and the plateau is the imprint of this correct selection rather than an allocator-side limitation. 
The residual gap to the upper bound, which integrates all eigenmodes including the weak ones that RMWF excludes, widens only gradually with $N$.

In contrast, RWF degrades almost immediately past $N = 8$: ranking coefficients by variance alone under a Gaussian model, it cannot concentrate the scarce budget where it is needed, and its sum GMI falls monotonically thereafter.
Equal-Bit, on the other hand, keeps climbing until $N=14$ before degrading.
As Fig.~\ref{fig:8vs16} shows, Equal-Bit at $N=16$ is not fully saturated at $C_{\sf tot}=896$. 
Consequently, its peak in Fig.~\ref{fig:sweep} emerges at $N=14$, where $896$ bits suffice to hit a per-$N$ ceiling that stays below the upper bound due to top-half selection.
Beyond $N=14$, $b=2C_{\sf tot}/(MN)$ becomes so small that even the dominant coefficients are quantized too coarsely.
Lacking any mechanism to isolate the informative ones, Equal-Bit lets this distortion outweigh the spatial gain, losing about $24.52$~bits/symbol from its $N=14$ peak by $N=32$.

These behaviors carry a clear operational message: only RMWF turns additional antennas into a sustained benefit under a fixed fronthaul budget.
RMWF holds its GMI on a near-flat plateau as $N$ grows, whereas RWF and Equal-Bit convert antenna over-provisioning into a GMI loss. 
In short, RMWF can scale the antenna count while the competing schemes cannot.
This scalability is an increasingly consequential property as antenna counts grow faster than fronthaul capacity in modern C-RAN deployments.

\subsection{CDF of Per-User GMI}

\begin{figure}[t]     
\centerline{\resizebox{0.8\columnwidth}{!}{\includegraphics{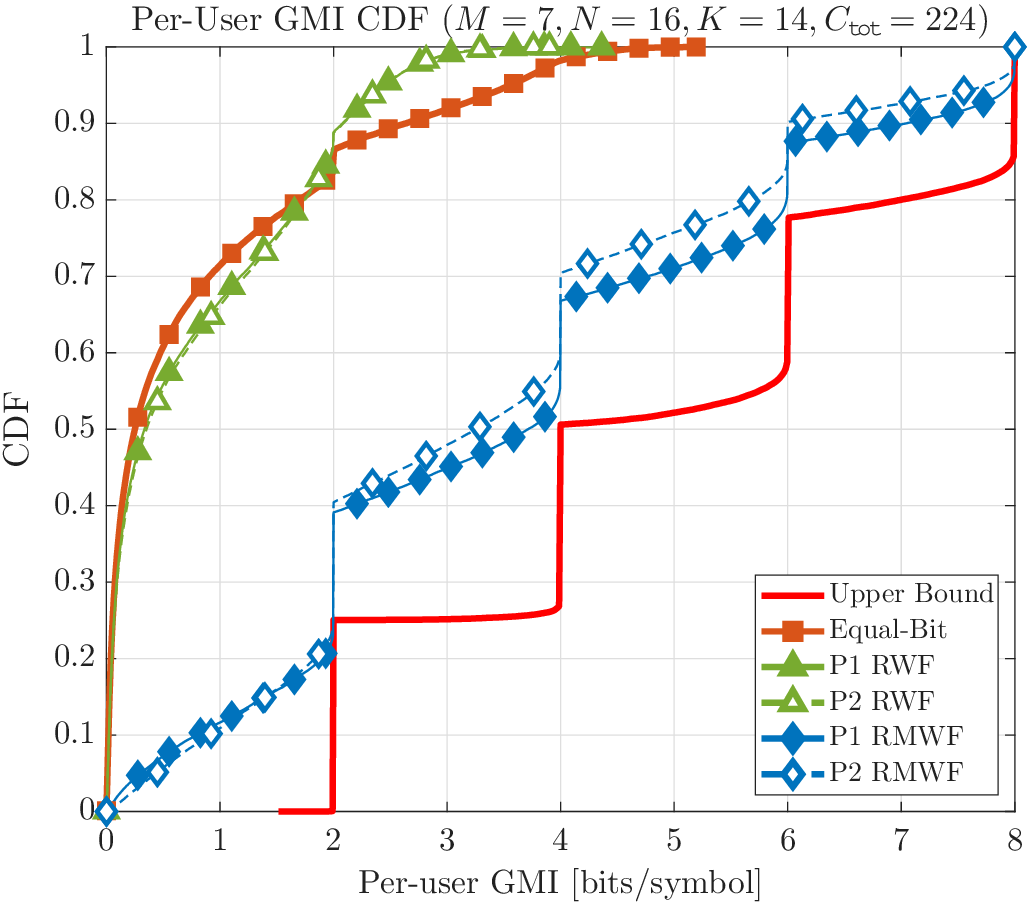}}}
\caption{CDF of per-user GMI for the sum (P1) and per-RU (P2) constraints;
$N = 16$ and $C_{\sf tot} = 224$.}
\label{fig:cdf}
\end{figure}

Fig.~\ref{fig:cdf} depicts the cumulative distribution function (CDF) of per-user GMI under the sum (P1) and per-RU (P2) fronthaul constraints. 
As shown in Fig.~\ref{fig:cdf}, under both Equal-Bit and RWF, approximately $50\%$ of the per-user GMI values fall below $0.35$~bits/symbol, and the vast majority remain below $4$~bits/symbol. 
This indicates that adaptive but Gaussian-based allocation provides little per-user benefit over Equal-Bit at this tight budget.

In contrast, the proposed RMWF achieves substantially higher per-user GMI. 
Under the sum constraint (P1), $50\%$ of the users obtain at least $3.71$~bits/symbol, and the $80$th percentile reaches $5.99$~bits/symbol. 
A similar trend is observed under the per-RU constraint (P2), where $50\%$ of the users obtain at least $3.25$~bits/symbol and the $80$th percentile reaches $5.68$~bits/symbol.

These results indicate that, regardless of whether a sum or per-RU fronthaul constraint is imposed, directly maximizing the finite-alphabet GMI is far more effective than minimizing reconstruction distortion under a Gaussian assumption.
This advantage manifests not only in the average sum GMI but also in the per-user GMI distribution, confirming that GMI-driven, finite-alphabet-aware bit allocation is critical from both system-level and user-level perspectives.
Furthermore, the RMWF advantage is somewhat larger under the sum constraint, since RMWF redistributes bits across RUs to lift more users.
The Gaussian-based RWF fails to exploit this flexibility, as its per-user distribution stays essentially unchanged between the two constraints.

\subsection{Convergence}
\label{subsec:convergence}

\begin{figure}[t]     
\centerline{\resizebox{0.8\columnwidth}{!}{\includegraphics{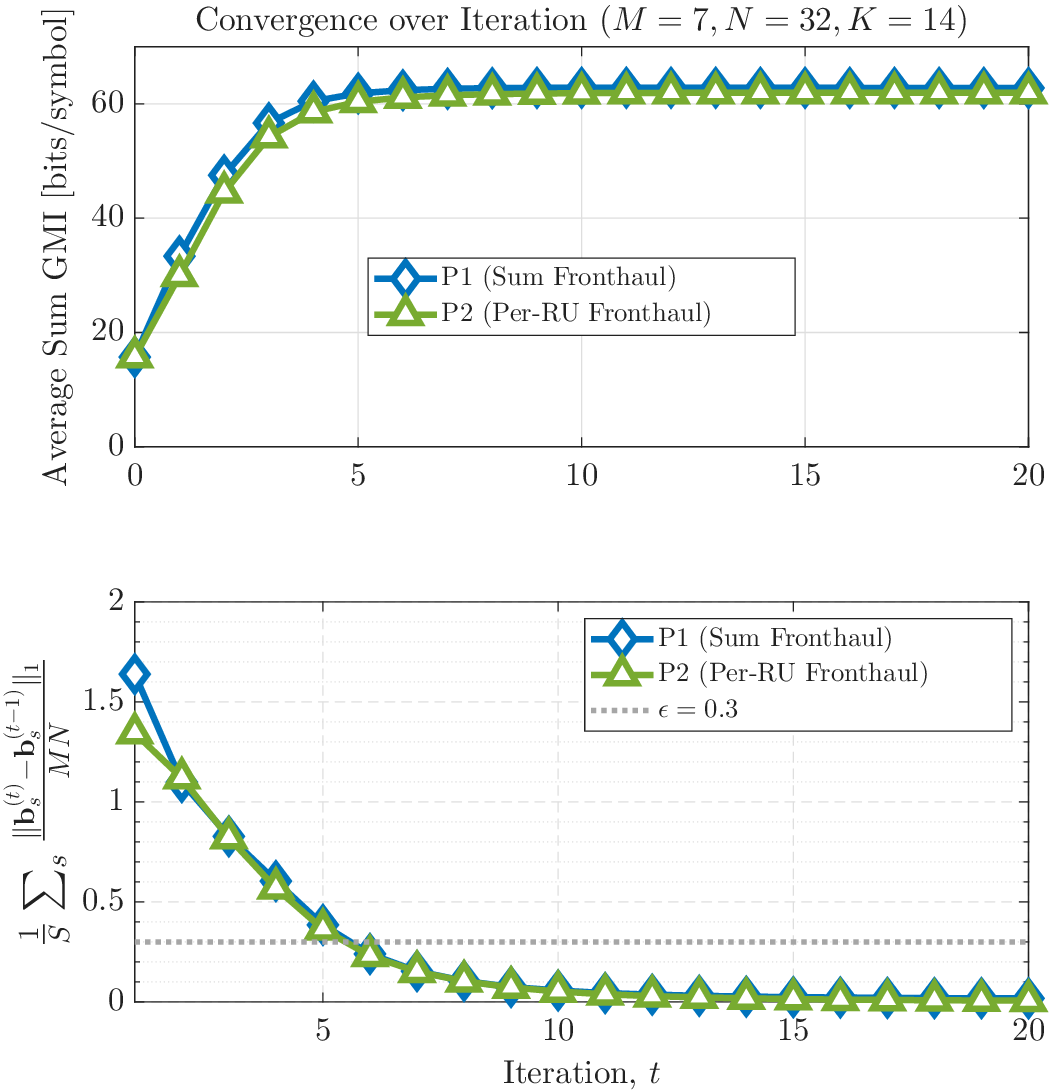}}}
\caption{Average sum GMI versus iteration index $t$ under $N = 32$ and $C_{\sf tot} = 448$, averaged over $S = 5000$ channel snapshots. 
Both Algorithm~\ref{alg:rmwf_p1} (P1) and Algorithm~\ref{alg:rmwf_p2} (P2) rise monotonically from the uniform initialization and reach a plateau within $T = 5$--$8$ iterations. 
The lower panel shows the per-element average bit-iterate change $\Delta^{(t)}$ defined in \eqref{eq:bit_change}; the stopping tolerance $\epsilon = 0.3$ is indicated, and the small saturated-regime residual visible after $t \approx 10$ does not affect the rounded bit allocation.}
\label{fig:GMIconv}
\end{figure}

A formal convergence proof is not pursued here; we instead report empirical evidence.
Fig.~\ref{fig:GMIconv} reports the convergence behavior over iterations. Its upper panel shows that the average sum GMI rises monotonically from the uniform initialization and plateaus within $T=5$--$8$ iterations for both (P1) and (P2).
The lower panel shows the per-element average bit change
\begin{align}
\Delta^{(t)} \triangleq \frac{1}{S}\sum_{s=1}^{S}
\frac{\|\mathbf{b}_s^{(t)} - \mathbf{b}_s^{(t-1)}\|_1}{MN},
\label{eq:bit_change}
\end{align}
averaged over $S=5000$ snapshots; 
the residual decays rapidly.
In implementation, we use the per-element-average stopping rule
$\Delta^{(t)}<\epsilon$ rather than $\|\cdot\|_\infty$ in Algorithms~\ref{alg:rmwf_p1}--\ref{alg:rmwf_p2}: norms on $\mathbb{R}^{MN}$ are equivalent, and the per-element form admits a size-independent reading, namely stabilization to within $\epsilon$ bits per coefficient on average. 
Because the real-valued allocation is subsequently rounded to even integers, sub-bit refinements have no effect on the quantized output, motivating $\epsilon=0.3$, below the rounding step. As the lower panel of Fig.~\ref{fig:GMIconv} shows, this threshold is met within $T\approx 6$ iterations for both (P1) and (P2), consistent with the GMI plateau.



\section{Conclusion}
\label{sec:conclusion}

We addressed fronthaul bit allocation for uplink C-RAN under finite-alphabet inputs, taking the post-LMMSE finite-alphabet GMI at the CU as the system objective.
Applied on the bit allocation side, the I-MMSE relation yields a fixed-point iteration termed RMWF.
Its converged allocation decomposes into a per-coefficient vessel height, a shared water level, and a finite-alphabet mercury level that captures decoder saturation.
Notably, the same structure handles both sum and per-RU fronthaul regimes, using either a single global water level or $M$ RU-local levels to govern the redistribution.
RMWF generalizes classical RWF through this saturation correction and completes a four-way correspondence between Gaussian and finite-alphabet inputs across the power and bit allocation sides.
Numerical results on a 3GPP 3D UMi C-RAN with mixed modulation orders confirmed that RMWF substantially outperforms Equal-Bit and Gaussian baselines and remains robust under antenna scaling.

\bibliographystyle{IEEEtran}
\bibliography{ref_nfc}

@string{asilomar="Proc. of IEEE Asilomar Conf. on Signals, Syst. and Comput."}

@string{jsac="IEEE J. Sel. Areas Commun."}

@string{twc="IEEE Trans. Wireless Commun."}

@string{tit="IEEE Trans. Inf. Theory"}

@string{access="IEEE Access"}

@string{network="IEEE Network Mag."}

@string{bstj="The Bell Sys. Tech. Jour."}

@article{peng2015fronthaul,
  author  = {Peng, Mugen and Wang, Chonggang and Lau, Vincent and Poor, H. Vincent},
  title   = {Fronthaul-Constrained Cloud Radio Access Networks: Insights and Challenges},
  journal = {IEEE Wireless Commun.},
  volume  = {22},
  number  = {2},
  pages   = {152--160},
  year    = {2015}
}

@article{park2014fronthaul,
  author  = {Park, Seok-Hwan and Simeone, Osvaldo and Sahin, Onur and Shamai (Shitz), Shlomo},
  title   = {Fronthaul Compression for Cloud Radio Access Networks: Signal Processing Advances Inspired by Network Information Theory},
  journal = {IEEE Signal Process. Mag.},
  volume  = {31},
  number  = {6},
  pages   = {69--79},
  year    = {2014}
}

@article{zhou2014optimized,
  author  = {Zhou, Yuhan and Yu, Wei},
  title   = {Optimized Backhaul Compression for Uplink Cloud Radio Access Network},
  journal = {IEEE J. Sel. Areas Commun.},
  volume  = {32},
  number  = {6},
  pages   = {1295--1307},
  year    = {2014}
}

@article{sanderovich2009uplink,
  author  = {Sanderovich, Amichai and Somekh, Oren and Poor, H. Vincent and Shamai (Shitz), Shlomo},
  title   = {Uplink Macro Diversity of Limited Backhaul Cellular Network},
  journal = {IEEE Trans. Inf. Theory},
  volume  = {55},
  number  = {8},
  pages   = {3457--3478},
  year    = {2009}
}

@article{park2017mixed,
  author  = {Park, Jeonghun and Park, Sungwoo and Yazdan, Ali and Heath Jr., Robert W.},
  title   = {Optimization of Mixed-{ADC} Multi-Antenna Systems for Cloud-{RAN} Deployments},
  journal = {IEEE Trans. Commun.},
  volume  = {65},
  number  = {9},
  pages   = {3962--3975},
  year    = {2017}
}

@article{gesbert2010multi,
  author  = {Gesbert, David and Hanly, Stephen and Huang, Howard and Shamai (Shitz), Shlomo and Simeone, Osvaldo and Yu, Wei},
  title   = {Multi-Cell {MIMO} Cooperative Networks: A New Look at Interference},
  journal = {IEEE J. Sel. Areas Commun.},
  volume  = {28},
  number  = {9},
  pages   = {1380--1408},
  year    = {2010}
}

@article{park2013robust,
  author={Park, Seok-Hwan and Simeone, Osvaldo and Sahin, Onur and Shamai, Shlomo},
  title={Robust and Efficient Distributed Compression for Cloud Radio Access Networks}, 
  journal={IEEE Trans. Veh. Technol.},   
  volume={62},
  number={2},
  pages={692--703},
  year={2013}
  }

@article{park2013joint,
  author  = {Park, Seok-Hwan and Simeone, Osvaldo and Sahin, Onur and Shamai (Shitz), Shlomo},
  title   = {Joint Precoding and Multivariate Backhaul Compression for the Downlink of Cloud Radio Access Networks},
  journal = {IEEE Trans. Signal Process.},
  volume  = {61},
  number  = {22},
  pages   = {5646--5658},
  year    = {2013}
}

@article{rao2015distributed,
  author  = {Rao, Xiongbin and Lau, Vincent K. N.},
  title   = {Distributed Fronthaul Compression and Joint Signal Recovery in Cloud-{RAN}},
  journal = {IEEE Trans. Signal Process.},
  volume  = {63},
  number  = {4},
  pages   = {1056--1065},
  year    = {2015}
}

@article{qiao2024meta,
  author={Qiao, Ruihua and Jiang, Tao and Yu, Wei},
  journal=twc, 
  title={Meta-Learning-Based Fronthaul Compression for Cloud Radio Access Networks}, 
  year={2024},
  volume={23},
  number={9},
  pages={11{}015-11{}029},
  doi={10.1109/TWC.2024.3378186}}

@article{bian2025towards,
  author  = {Bian, Chenghong and Shao, Yulin and G\"{u}nd\"{u}z, Deniz},
  title   = {Towards {AI}-Native Fronthaul: Neural Compression for Next{G} Cloud {RAN}},
  journal = {arXiv preprint arXiv:2506.06925},
  year    = {2025}
}

@inproceedings{balle2018variational,
  author    = {Ball\'{e}, Johannes and Minnen, David and Singh, Saurabh and Hwang, Sung Jin and Johnston, Nick},
  title     = {Variational Image Compression with a Scale Hyperprior},
  booktitle = {Int. Conf. Learn. Represent. (ICLR)},
  year      = {2018}
}

@article{zhou2016fronthaul,
  author  = {Zhou, Yuhan and Yu, Wei},
  title   = {Fronthaul Compression and Transmit Beamforming Optimization for Multi-Antenna Uplink {C-RAN}},
  journal = {IEEE Trans. Signal Process.},
  volume  = {64},
  number  = {16},
  pages   = {4138--4151},
  year    = {2016}
}

@article{han2022sparse,
  author  = {Han, Deokhwan and Park, Jeonghun and Park, Seok-Hwan and Lee, Namyoon},
  title   = {Sparse Joint Transmission for Cloud Radio Access Networks with Limited Fronthaul Capacity},
  journal = {IEEE Trans. Wireless Commun.},
  volume  = {21},
  number  = {5},
  pages   = {3395--3408},
  year    = {2022}
}

@article{yu2021deep,
  author  = {Yu, Daesung and Lee, Hoon and Park, Seok-Hwan and Hong, Seung-Eun},
  title   = {Deep Learning Methods for Joint Optimization of Beamforming and Fronthaul Quantization in Cloud Radio Access Networks},
  journal = {IEEE Wireless Commun. Lett.},
  volume  = {10},
  number  = {10},
  pages   = {2180--2184},
  year    = {2021}
}

@article{zhou2016optimal,
  author  = {Zhou, Yuhan and Xu, Yinfei and Yu, Wei and Chen, Jun},
  title   = {On the Optimal Fronthaul Compression and Decoding Strategies for Uplink Cloud Radio Access Networks},
  journal = {IEEE Trans. Inf. Theory},
  volume  = {62},
  number  = {12},
  pages   = {7402--7418},
  year    = {2016}
}

@article{liu2021duality,
  author  = {Liu, Liang and Liu, Ya-Feng and Patil, Pratik and Yu, Wei},
  title   = {Uplink-Downlink Duality Between Multiple-Access and Broadcast Channels with Compressing Relays},
  journal = {IEEE Trans. Inf. Theory},
  volume  = {67},
  number  = {11},
  pages   = {7304--7337},
  year    = {2021}
}

@article{patil2019generalized,
  author  = {Patil, Pratik and Yu, Wei},
  title   = {Generalized Compression Strategy for the Downlink Cloud Radio Access Network},
  journal = {IEEE Trans. Inf. Theory},
  volume  = {65},
  number  = {10},
  pages   = {6766--6780},
  year    = {2019}
}

@article{patil2018hybrid,
  author  = {Patil, Pratik and Dai, Binbin and Yu, Wei},
  title   = {Hybrid Data-Sharing and Compression Strategy for Downlink Cloud Radio Access Network},
  journal = {IEEE Trans. Commun.},
  volume  = {66},
  number  = {11},
  pages   = {5370--5384},
  year    = {2018}
}

@article{liu2019twotimescale,
  author  = {Liu, An and Chen, Xihan and Yu, Wei and Lau, Vincent K. N. and Zhao, Min-Jian},
  title   = {Two-Timescale Hybrid Compression and Forward for Massive {MIMO} Aided {C-RAN}},
  journal = {IEEE Trans. Signal Process.},
  volume  = {67},
  number  = {9},
  pages   = {2484--2498},
  year    = {2019}
}

@article{sohrabi2022learning,
  author  = {Sohrabi, Foad and Jiang, Tao and Yu, Wei},
  title   = {Learning Progressive Distributed Compression Strategies from Local Channel State Information},
  journal = {IEEE J. Sel. Topics Signal Process.},
  volume  = {16},
  number  = {3},
  pages   = {573--584},
  year    = {2022}
}

@book{CoverThomas,
  author    = {Cover, Thomas M. and Thomas, Joy A.},
  title     = {Elements of Information Theory},
  publisher = {Wiley-Interscience},
  address   = {Hoboken, NJ, USA},
  edition   = {2nd},
  year      = {2006}
}

@book{GershoGray,
  author    = {Gersho, Allen and Gray, Robert M.},
  title     = {Vector Quantization and Signal Compression},
  publisher = {Kluwer Academic Publishers},
  address   = {Norwell, MA, USA},
  series    = {The Springer International Series in Engineering and Computer Science},
  volume    = {159},
  year      = {1992}
}

@article{GuoShamaiVerdu,
  author  = {Guo, Dongning and Shamai (Shitz), Shlomo and Verd{\'u}, Sergio},
  title   = {Mutual Information and Minimum Mean-Square Error in {G}aussian Channels},
  journal = {IEEE Trans. Inf. Theory},
  volume  = {51},
  number  = {4},
  pages   = {1261--1282},
  year    = {2005},
  doi     = {10.1109/TIT.2005.844072}
}

@ARTICLE{LozanoTulinoVerdu,
  title={Optimum power allocation for parallel {G}aussian channels with arbitrary input distributions},
  author={Lozano, Angel and Tulino, Antonia M and Verd{\'u}, Sergio},
  journal={IEEE Trans. Inf. Theory},
  volume={52},
  number={7},
  pages={3033--3051},
  year={2006},
  publisher={IEEE}
}

@article{merhav1994information,
  author  = {N. Merhav and G. Kaplan and A. Lapidoth and S. {Shamai (Shitz)}},
  title   = {On information rates for mismatched decoders},
  journal = {{IEEE} Trans. Inf. Theory},
  volume  = {40},
  number  = {6},
  pages   = {1953--1967},
  year    = {1994}
}

@article{checko2015cloud,
  author  = {A. Checko and H. L. Christiansen and Y. Yan and L. Scolari and G. Kardaras and M. S. Berger and L. Dittmann},
  title   = {Cloud {RAN} for mobile networks---{A} technology overview},
  journal = {{IEEE} Commun. Surveys Tuts.},
  volume  = {17},
  number  = {1},
  pages   = {405--426},
  year    = {2015}
}

@book{boyd2004convex,
  author    = {S. Boyd and L. Vandenberghe},
  title     = {Convex Optimization},
  publisher = {Cambridge University Press},
  address   = {Cambridge, U.K.},
  year      = {2004}
}

@article{bennett1948spectra,
  title={Spectra of quantized signals},
  author={Bennett, William Ralph},
  journal=bstj,
  volume={27},
  number={3},
  pages={446--472},
  year={1948},
  publisher={Nokia Bell Labs}
}

@article{gishpierce1968asymptotically,
  author={Gish, H. and Pierce, J.},
  journal=tit, 
  title={Asymptotically efficient quantizing}, 
  year={1968},
  volume={14},
  number={5},
  pages={676-683},
  doi={10.1109/TIT.1968.1054193}}

@article{liu2015optimized,
  author={Liu, Liang and Zhang, Rui},
  title={Optimized uplink transmission in multi-antenna {C-RAN} with spatial compression and forward},  
  journal={IEEE Trans. Signal Process.},
  volume={63},
  number={19},
  pages={5083--5095},
  year={2015},
  publisher={IEEE}
}

@inproceedings{chene2024distributed,
  author={Ch{\^e}ne, Thomas and Othman, Ghaya Rekaya-Ben and Damen, Oussama},
  title={Distributed Decoding Scheme for Uplink {C-RAN} System with Limited Backhaul Capacity},  
  booktitle={2024 58th Asilomar Conference on Signals, Systems, and Computers},
  pages={1438--1442},
  year={2024},
  organization={IEEE}
}

@article{zhang2021quantization,
  author={Zhang, Chao and Askri, Aymen and Othman, Ghaya Rekaya-Ben and Wang, Li},
  title={Quantization-Aware Processing for Massive {MIMO} Uplink Cloud {RAN}},
  journal={IEEE Commun. Lett.},
  volume={26},
  number={2},
  pages={468--472},
  year={2022},
  publisher={IEEE}
}

@misc{oran2020wp,
  author = {{O-RAN Alliance}},
  title = {{O-RAN: Towards an Open and Smart RAN}},
  howpublished = {White Paper},
  year = {2018}
}

@inproceedings{wiffen2021distributed,
  title={Distributed dimension reduction for distributed massive {MIMO} {C-RAN} with finite fronthaul capacity},
  author={Wiffen, Fred and Chin, Woon Hau and Doufexi, Angela},
  booktitle={2021 55th Asilomar Conference on Signals, Systems, and Computers},
  pages={1228--1236},
  year={2021},
  organization={IEEE}
}

@article{wiffen2020dimension,
  author={Wiffen, Fred and Bocus, Mohammud Z and Chin, Woon Hau and Doufexi, Angela and Beach, Mark},
  title={Dimension reduction-based signal compression for uplink distributed {MIMO} {C-RAN} with limited fronthaul capacity},  
  journal={arXiv preprint arXiv:2005.12894},
  year={2020}
}

@article{ganti2000mismatched,
author={Ganti, Anand and Lapidoth, Amos and Telatar, I Emre},
title={Mismatched decoding revisited: General alphabets, channels with memory, and the wide-band limit},
journal={{IEEE} Trans. Inf. Theory},
volume={46},
number={7},
pages={2315--2328},
year={2000},
publisher={IEEE}
}

@ARTICLE{bruno2026mismatch,
  author={Zhang, Sibo and Clerckx, Bruno},
  journal=jsac, 
  title={Optimal and Suboptimal Decoders Under Finite-Alphabet Interference: {A} Mismatched Decoding Perspective}, 
  year={2026},
  volume={44},
  number={},
  pages={2779-2793},
  doi={10.1109/JSAC.2025.3644829}}

@article{schuchman1964dither,
  title={Dither signals and their effect on quantization noise},
  author={Schuchman, Leonard},
  journal={IEEE Trans. Commun. Technol.},
  volume={12},
  number={4},
  pages={162--165},
  year={1964},
  publisher={IEEE}
}

@inproceedings{carmon2012disproof,
  title={Disproof of the {S}hamai--{L}aroia conjecture},
  author={Carmon, Yair and Shamai, Shlomo and Weissman, Tsachy},
  booktitle={2012 IEEE 27th Convention of Electrical and Electronics Engineers in Israel},
  pages={1--4},
  year={2012},
  organization={IEEE}
}

\end{document}